\documentclass[11pt]{article}
\usepackage[colorlinks]{hyperref}
\hypersetup{linkcolor=cyan,filecolor=cyan,citecolor=cyan,urlcolor=cyan}
\usepackage{xspace}
\usepackage{thm-restate,color,xcolor}
\usepackage{amsmath,amssymb,bbm,amsthm,algorithmicx,algorithm,algpseudocode}
\usepackage{fullpage}
\usepackage{boxedminipage}
\usepackage[sc]{mathpazo}
\usepackage{amsmath}
\usepackage{mathtools}
\usepackage{framed}
\usepackage[framemethod=tikz]{mdframed}
\usepackage{titlesec}
\usepackage{lipsum}%
\usepackage{cleveref,aliascnt}
\usepackage{tikz}
\usepackage{wrapfig}

\algblock{ParFor}{EndParFor}
\algnewcommand\algorithmicparfor{\textbf{parfor}}
\algnewcommand\algorithmicpardo{\textbf{do}}
\algnewcommand\algorithmicendparfor{\textbf{end\ parfor}}
\algrenewtext{ParFor}[1]{\algorithmicparfor\ #1\ \algorithmicpardo}
\algrenewtext{EndParFor}{\algorithmicendparfor}

\usetikzlibrary{shapes.geometric}
\usetikzlibrary{arrows}
\usetikzlibrary{arrows.meta}
\usetikzlibrary{patterns}
\usetikzlibrary{shapes.misc}

\newcommand{\dist}{\mbox{\rm dist}}

\newtheorem{theorem}{Theorem}[section]
\newtheorem{lemma}[theorem]{Lemma}
\newtheorem{meta-theorem}[theorem]{Meta-Theorem}
\newtheorem{claim}[theorem]{Claim}

\newtheorem{corollary}[theorem]{Corollary}

\newtheorem{observation}[theorem]{Observation}
\newtheorem{definition}[theorem]{Definition}

\newcommand{\Oish}{\widetilde{O}}

\def\LastE{\mbox{\tt LastE}}

\def\Shortcut{\mbox{\tt Shortcut}}
\def\VertexFTSpanner{\mbox{\tt VFTSpanner}}

\usepackage{geometry}
\begin{document}
\date{}

\title{\~{O}ptimal Vertex Fault-Tolerant Spanners in \~{O}ptimal Time:\\ Sequential, Distributed and Parallel}
\author{	
Merav Parter \thanks{This project is partially funded by the European Research Council (ERC) under the European Union’s Horizon 2020 research and innovation programme, grant agreement No. 949083, and the Israeli Science Foundation (ISF), grant 2084/18.}\\
        \small Weizmann Institute \\
        \small merav.parter@weizmann.ac.il
}

\maketitle

\begin{abstract}
We (nearly) settle the time complexity for computing vertex fault-tolerant (VFT) spanners with optimal sparsity (up to polylogarithmic factors). VFT spanners are sparse subgraphs that preserve distance information, up to a small multiplicative stretch, in the presence of vertex failures. These structures were introduced by [Chechik et al., STOC 2009] and have received a lot of attention since then. 

Recent work provided algorithms for computing VFT spanners with \emph{optimal} sparsity but in exponential runtime. 
The first polynomial time algorithms for these structures have been given by [Bodwin, Dinitz and Robelle, SODA 2021]. Their algorithms, as all other prior algorithms, are greedy and thus inherently sequential.
We provide algorithms for computing nearly optimal $f$-VFT spanners for any $n$-vertex $m$-edge graph, with near optimal running time in several computational models:

\begin{itemize}
\item A randomized sequential algorithm with a runtime of $\widetilde{O}(m)$ (i.e., independent in the number of faults $f$). The state-of-the-art time bound is $\widetilde{O}(f^{1-1/k}\cdot n^{2+1/k}+f^2 m)$ by [Bodwin, Dinitz and Robelle, SODA 2021]. 

\item  A distributed congest algorithm of $\widetilde{O}(1)$ rounds. Improving upon [Dinitz and Robelle, PODC 2020] that obtained FT spanners with near-optimal sparsity in $\widetilde{O}(f^{2})$ rounds. 

\item  A PRAM (CRCW) algorithm with $\widetilde{O}(m)$ work and $\widetilde{O}(1)$ depth. Prior bounds implied by [Dinitz and Krauthgamer, PODC 2011] obtained sub-optimal FT spanners using $\widetilde{O}(f^3m)$ work and $\widetilde{O}(f^3)$ depth. 

\end{itemize}

An immediate corollary provides the first nearly-optimal PRAM algorithm for computing nearly optimal $\lambda$-\emph{vertex} connectivity certificates using polylogarithmic depth and near-linear work. This improves the state-of-the-art parallel bounds of $\widetilde{O}(1)$ depth and $O(\lambda m)$ work, by [Karger and Motwani, STOC'93].
\end{abstract}
\newpage
\tableofcontents

\newpage

\section{Introduction}
This paper is concerned with time-efficient algorithms for computing optimal vertex fault-tolerant (FT) spanners, i.e., of optimal (or nearly optimal) sparsity. Graph spanners introduced by Peleg and Sch{\"{a}}ffer (\cite{PelegS:89,PelegU:89}) are sparse subgraphs that preserve the shortest path metric, up to a small multiplicative stretch. A landmark result of Alth\"ofer et al.~\cite{AlthoferDDJS:93} proved that for any integer $k \geq 1$, every $n$-vertex graph $G=(V, E)$ has a $(2k-1)$-spanner $H \subseteq G$ with $O(n^{1+1/k})$ edges. This tradeoff is believed to be tight by the girth conjecture of Erd\H{o}s \cite{erdHos1964extremal}.  Spanners have a wide-range of applications in routing \cite{PelegU:89-routing}, synchronizers \cite{awerbuch1990network}, distance oracles \cite{thorup2005approximate}, graph sparsifiers \cite{kapralov2012spectral} and preconditioning of linear systems \cite{elkin2008lower}.  

\vspace{-5pt}
\subsection{Vertex Fault Tolerant Spanners}\vspace{-3pt}
Many of the applications of spanners arise in the context of distributed networks which are inherently prune to failures of edges and vertices. It is then desirable to obtain robust spanners that maintain their functionality in the presence of faults. \emph{Fault-tolerant (FT) spanners} provide this guarantee by containing a spanner in $G \setminus F$ for any possible small subset $F \subset V$. These structures were introduced in the context of geometric graphs by Levcopoulos et al. \cite{levcopoulos1998efficient}, Czumaj and Zhao \cite{czumaj2004fault}, and later on, for general graphs by Chechik et al. \cite{ChechikLPR:10}. 

\begin{definition}[Vertex $f$-FT $t$-Spanners]
For a given $n$-vertex (possibly weighted) graph $G=(V,E)$, a subgraph $H \subseteq G$ is a vertex $f$-FT $t$-spanner for $G$ if 
$$\dist_{H \setminus F}(u,v)\leq t \cdot \dist_{G \setminus F}(u,v), \forall u,v, F \in V \times V \times V^{\leq f}.$$
\end{definition}
FT-spanners have attracted a lot of attention since their introduction. The initial efforts went into pinning the existentially optimal size bounds as a function of $n,t$ and $f$. The focus of many of the recent works is in the computational time aspects for computing these structures.

\vspace{-7pt}
\paragraph{The quest for FT spanners with optimal size.} The first results on FT spanners were given by Chechik et al. \cite{ChechikLPR:10} that presented an ingenious extension of the Thorup-Zwick algorithm \cite{thorup2005approximate} for obtaining $f$-FT $(2k-1)$-spanners with $\widetilde{O}(k^f n^{1+1/k})$ edges\footnote{The notation $\widetilde{O}(\cdot)$ hides polylogarithmic terms in $n$.}. In a subsequent work, Dinitz and Krauthgamer \cite{DinitzK:11} presented a general simulation result that translates any (fault-free) $(2k-1)$ spanner algorithm into an $f$-FT $(2k-1)$ spanner algorithm while paying an overhead of $\widetilde{O}(f^{2-1/k})$ in the spanner size. The quest for optimal FT-spanners has been marked by Bodwin et al. \cite{BDPW18} and Bodwin and Patel \cite{BP19}. By analyzing the output spanner of exponential-time greedy algorithms, \cite{BDPW18,BP19} obtained the desired size bound of $O(f^{1-1/k}n^{1+1/k})$ edges. These bounds were also shown in  \cite{BDPW18} to be existentially tight conditioned on the Erd\H{o}s girth conjecture. Obtaining the same (or similar) size bounds in polynomial time was mentioned as an important open question in \cite{BDPW18,BP19}.

\vspace{-10pt}
\paragraph{The quest for optimal FT spanners in optimal time.} Unlike the (fault-free) greedy algorithm of Alth\"ofer et al.~\cite{AlthoferDDJS:93}, the naive implementation of the FT-greedy algorithms of \cite{BDPW18,BP19} requires exponential time, in the number of faults $f$. Dinitz and Robelle \cite{DR20} presented an elegant implementation of these greedy algorithms to run in $O(k \cdot f^{2-1/k}n^{1+1/k}\cdot m)$ time, and with nearly optimal sparsity of $O(k f^{1-1/k}n^{1+1/k})$ edges. In a more recent work, Bodwin, Dinitz and Robelle \cite{BodwinDR21} obtained truly optimal spanners in time  $\widetilde{O}(f^{1-1/k}\cdot n^{2+1/k}+m f^2)$. Their algorithm, as all prior algorithms for optimal VFT-spanners, is greedy (with some slack), and its efficient implementation exploits the blocking set technique, first introduced in \cite{BP19}. See Table \ref{tab:priorwork} for a summary of the state-of-the art bounds for VFT-spanners. 
As (non-faulty) spanners can be computed in nearly linear time (e.g., the Baswana-Sen algorithm \cite{BaswanaS:07}), in their paper, Bodwin, Dinitz and Robelle ask:
\vspace{-4pt}
\begin{quote}\cite{BodwinDR21}
\emph{Is it possible to compute optimal-size fault-tolerant spanners in time $\widetilde{O}(m)$?}
\end{quote}
\vspace{-4pt}
In this work, we answer this question in the affirmative up to paying an extra poly-logarithmic term in the size bound of the optimal FT spanners.
%

\begin{table*}[t]
\begin{center}

\begin{tabular}{lllcl}
\textbf{Spanner Size} & \textbf{Sequential time} & \textbf{Greedy?} & \textbf{Distributed time} & \textbf{Citation} \\
        $\Oish \left(k^{O(f)} \cdot n^{1+1/k} \right)$  & $\Oish \left(k^{O(f)} \cdot n^{3+1/k} \right)$ & & NA & \cite{ChechikLPR:10} \\
        $\Oish \left(f^{2 - 1/k} \cdot n^{1+1/k} \right)$ & {\color{blue} $\Oish\left( f^{2 - 2/k} \cdot m \right)$} & & {\color{purple}$\Oish(f^2)$} & \cite{DinitzK:11,DR20}\\
        $O \left( \exp(k) f^{1 - 1/k} \cdot n^{1+1/k} \right) $ & $O\left( \exp(k) \cdot m n^{O(f)} \right)$ & \checkmark{} & NA &\cite{BDPW18} \\
        \color{red} $O \left( f^{1 - 1/k} \cdot n^{1+1/k} \right)$ & $O\left( mn^{O(f)} \right)$ & \checkmark{} & NA & \cite{BP19} \\
       \color{red} $O \left( k f^{1 - 1/k} \cdot n^{1+1/k} \right)$ & \color{blue}{$\Oish\left( f^{2 - 1/k} \cdot m n^{1+1/k} \right)$} & (\checkmark{}) & NA & \cite{DR20} \\
        \color{red} $O \left( f^{1 - 1/k} \cdot n^{1+1/k} \right)$ & \color{blue}$\Oish\left( f^{1 - 1/k} n^{2+1/k} + mf^2\right)$ & (\checkmark{}) & NA &\cite{BodwinDR21}\\
				
        \color{red} $\Oish \left( f^{1 - 1/k} \cdot n^{1+1/k} \right)$ & \color{blue}$\Oish(m)$ &  & {\color{purple} $\Oish(1)$} &\textbf{(this paper)}
    \end{tabular}
\caption{\label{tab:priorwork} Prior work on $f$-VFT $(2k-1)$-spanners of weighted input graphs on $n$ vertices and $m$ edges (based on Table 1 in \cite{BodwinDR21}).  Size bounds in red are (nearly) existentially optimal, and computational time in blue are polynomial.  The (\checkmark{}) entries indicate a greedy-like algorithm.}
\end{center}
\vspace{-15pt}
\end{table*}
\vspace{-7pt}

\paragraph{Distributed and Parallel Algorithms.} Previous work on optimal FT spanners has focused on their centralized (or sequential) construction. As all prior algorithms for these structures are greedy, it is unclear how to implement them in a distributed and parallel environments. Indeed, despite the fact that the key motivation for fault tolerant spanners comes from distributed networks, currently we are lacking time-efficient algorithms for optimal FT spanners. The known algorithms by Dinitz and Krauthgamer \cite{DinitzK:11} and Dinitz and Robelle \cite{DR20} provide FT-spanners with sub-optimal size of $\widetilde{O}(f^{2-1/k}n^{1+1/k})$ edges, and using $\widetilde{O}(f^{2-1/k})$ congest rounds \cite{Peleg:2000}. We note that even for the simpler setting of edge-FT spanners (resilient to $f$ edge faults), currently there are no local solutions, i.e., with $\widetilde{O}(1)$ congest rounds, even when settling for spanners with sub-optimal sparsity\footnote{The only known distributed algorithm for edge $f$-FT spanners (obtained by \cite{ChechikLPR:10}) in  runs in $\widetilde{O}(f)$ congest rounds and obtains spanners with $\widetilde{O}(fn^{1+1/k})$ edges.}. Altogether, the existing distributed constructions for FT spanners provide sub-optimal FT spanners while using $\Omega(f)$ number of rounds. This was also posed as an important problem\footnote{The problem was also posed in the Advances in Distributed Graph Algorithms (ADGA) workshop 2021 \cite{Dinitz20}.} in \cite{BodwinDR21}. 
\begin{quote}\cite{BodwinDR21}
\emph{The greedy algorithm is typically difficult to parallelize or to implement efficiently distributedly (particularly in the the presence of congestion). So there is the obvious question of computing optimal-size fault-tolerant spanners efficiently in these models.} 
\end{quote}
To our knowledge, the situation in the PRAM setting is similar. The only known algorithm, implicit by \cite{DinitzK:11} provides $f$-VFT $(2k-1)$ spanners with $\widetilde{O}(f^{2-1/k}n^{1+1/k})$ edges, $\widetilde{O}(f^{3-1/k})$ depth and $\widetilde{O}(f^{3-1/k}m)$ work.

\subsection{Connectivity Certificates} \vspace{-3pt}
A closely related graph structure for FT-spanners is the $\lambda$-vertex (or edge) connectivity certificates, which, roughly speaking, provide a succinct ``proof" for the $\lambda$ connectivity of the graph.  Formally, a $\lambda$-vertex (edge) connectivity certificate $H \subseteq G$ satisfies that $H$ $\lambda$-vertex (edge) connected iff $G$ is $\lambda$-vertex (edge) connected. Connectivity certificates play a key role in almost any minimum cut algorithm, a sample includes \cite{Matula93,CheriyanKT93,KargerM97,GhaffariK13,DagaHNS19,ForsterNYSY20,LiNPSY21}. 

Connectivity certificates were introduced by Nagamochi and Ibaraki \cite{NagamochiI92} that demonstrated the existence of $\lambda$-connectivity certificates with $\lambda(n-1)$ edges for every $n$-vertex undirected graph. A $\lambda$-connectivity certificate with  $O(\lambda n)$ edges is denoted\footnote{Computing certificates of minimum size is known to be NP-complete, and therefore we settle for constant approximation of $O(\lambda n)$ edges.} as \emph{sparse}.
\cite{NagamochiI92} also presented a linear time (sequential) algorithm for sparse vertex and edge certificates.  The basic meta-algorithm of \cite{NagamochiI92} is based upon the computation of $\lambda$ edge-disjoint maximal forests, computed sequentially in $\lambda$ iterations. A na\"ive implementation of this algorithm leads to an inherent linear dependency in the running time. This dependency is currently avoided only in the sequential setting.  

\vspace{-5pt}
\paragraph{PRAM and Distributed Algorithms.}  Due to their algorithmic importance, the computation of sparse connectivity certificates in the parallel and distributed settings has attracted much attention over the years. The state-of-the-art PRAM bounds for vertex connectivity certificates are given by Karger and Motwani \cite{KargerM97}  that obtain sparse $\lambda$-vertex certificates in logarithmic depth and $\widetilde{O}(\lambda m)$ work. See Table \ref{tab:priorwork-certificate} for a summary of the PRAM complexity for this problem. We note that for $\lambda$-\emph{edge} certificates (that preserve the \emph{edge} connectivity), Daga et al. \cite{DagaHNS19} provided only recently a PRAM algorithm with $\widetilde{O}(1)$ depth and total of $\widetilde{O}(m)$ work. No such algorithm is currently known for $\lambda$-\emph{vertex} certificates. 

In the context of distributed computing, \cite{Thurimella95} provided the first distributed algorithms for sparse certificates in the congest model. A na\"ive implementation of the sequential algorithm by \cite{NagamochiI92,CheriyanKT93} leads to an $\widetilde{O}(\lambda(\sqrt{n}+D))$-round algorithm, where $D$ is the diameter of the input graph $G$. For $\lambda$-\emph{edge} certificates, Daga et al. \cite{DagaHNS19} provided a $\widetilde{O}(\sqrt{n\lambda}+D)$-round algorithm for any $\lambda=O(n^{1-\epsilon})$. The author \cite{Par19} provided a $\widetilde{O}(\lambda)$-round algorithm for computing sparse $\lambda$-edge certificates. No such algorithms are known for sparse vertex certificates. In this work, we fill in the missing gap and provide truly local distributed and parallel algorithms for \emph{vertex} certificates with nearly optimal sparsity. These improved algorithms follow immediately by our FT spanner constructions, as explained next. 
 
\begin{table*}[t]
\begin{center}
\begin{tabular}{llcl}
        \textbf{Certificate Size} & \textbf{Work} & \textbf{Depth} & \textbf{Citation} \\
				
				$O(\lambda n)$  & $\Oish \left(\lambda n m\right)$ &  $\Oish \left(\lambda^2 \right)$ &  \cite{KhullerS89} \\
				
				$O(\lambda n)$  & $O(\lambda m)$ & $\Oish \left(\lambda \right)$ & \cite{CheriyanT91} \\

				$O(\lambda n)$  & $poly(n)$ & \color{blue} $\Oish \left(1 \right)$ & \cite{CheriyanT91} \\
				
        $\lambda(n-1)$  & $\Oish \left(\lambda m\right)$ & $\Oish \left(\lambda\right)$ & \cite{CheriyanKT93} \\
				
				$O(\lambda n)$  & $O(\lambda m)$ & \color{blue} $\Oish \left(1\right)$ & \cite{KargerM97} \\

				$O(\lambda^2 n)$  & $O(\lambda^2 m)$ & \color{blue} $\Oish \left(1\right)$ & \cite{DinitzK:11,DR20} \\
				
        $\Oish \left(\lambda n\right)$ & \color{red} $\Oish \left(m\right)$ & \color{blue}$\Oish(1)$  &\textbf{(this paper)}
 \end{tabular}
\caption{\label{tab:priorwork-certificate} Prior work on parallel computation of $\lambda$-vertex connectivity certificates. Depth bounds (resp., work bounds) in red (resp., blue) are (nearly) existentially optimal. The bounds of \cite{DinitzK:11,DR20} are implicit by their VFT-spanner constructions.}
\end{center}
\vspace{-10pt}
\end{table*}
\vspace{-10pt}
\paragraph{FT-Spanners are, in fact, Connectivity Certificates.} The author observed in \cite{Par19} that FT $(2k-1)$ spanners against $\lambda$ vertex (resp., edge) failures are, by definition, $\lambda$-connectivity certificates, for any $k$. The sparsest possible certificate is obtained by taking the optimal $\lambda$-FT $(2k-1)$ spanner for $k=O(\log n)$, which by \cite{BDPW18,BP19} contains $O(\lambda n)$ edges; hence a \emph{sparse} certificate. This connection immediately lead to $\widetilde{O}(\lambda)$-round distributed algorithms for nearly sparse $\lambda$-edge certificates. Using the algorithm of \cite{DR20}, one can obtain $\lambda$-\emph{vertex} certificates with $\widetilde{O}(\lambda^2 n)$ edges using $\widetilde{O}(\lambda)$ rounds. 
In this paper, we will use this connection to provide nearly-sparse vertex-certificates, i.e., with $\widetilde{O}(\lambda n)$ edges, and in nearly optimal parallel and distributed runtime (e.g., in $\widetilde{O}(1)$ congest rounds). 
\vspace{-6pt} \subsection{Our Contribution}\vspace{-4pt}
We provide a local (non-sequential) algorithm for nearly optimal FT-spanners whose locality parameter is independent in the number of faults $f$. Our algorithm is based on a new extension of the well-known Baswana-Sen algorithm \cite{BaswanaS:07} to the fault-tolerant setting. We show that this (meta) algorithm can be naturally implemented in nearly optimal time in the sequential, distributed and parallel settings. 
Our extension is based on a novel notion of \emph{Fault-Tolerant Clustering} that induces low-depth vertex-independent trees in $G$, in which each clustered vertex appear in $\Omega(f)$ many clusters. We hope that this notion will be useful for the design of additional fault-tolerant distance preserving structures. Our key result for $n$-vertex graphs with polynomial edge weights is:

\begin{mdframed}[hidealllines=true,backgroundcolor=gray!25]
\vspace{-5pt}
\begin{theorem}\label{thm:linear-time-opt-randomized}[\~{O}ptimal Spanner in \~{O}ptimal Sequential Time]
There is a randomized algorithm that given any $n$-vertex $m$-edge (possibly weighted) graph $G=(V,E,W)$, a stretch parameter $k$ and a fault bound $f \in [1,n]$ computes in $\widetilde{O}(m)$ time an $f$-VFT $(2k-1)$ spanner subgraph $H\subseteq G$ with $|E(H)|=\widetilde{O}(f^{1-1/k} n^{1+1/k})$ edges, w.h.p.
\end{theorem}
\end{mdframed}
This should be compared with the state-of-the-art runtime of $\widetilde{O}(f^{1-1/k}\cdot n^{2+1/k}+m f^2)$ by \cite{BodwinDR21}. The algorithm of \cite{BodwinDR21} has the benefit of computing truly optimal spanners with $O(f^{1-1/k} n^{1+1/k})$ edges\footnote{Their optimality is also unconditional on the girth conjecture.}. 
We note that even for edge failures all prior algorithms (even for sub-optimal FT-spanners) have an inherent dependency in $f$ in their running times. Specifically, the fastest $f$-EFT spanner algorithm by Chechik et al. \cite{ChechikLPR:10} runs in $\widetilde{O}(f m)$ time and provides spanners with $\widetilde{O}(fn^{1+1/k})$ edges. Our algorithm can be easily shown to provide also $f$-EFT spanners with $\widetilde{O}(f^{1-1/k}n^{1+1/k})$ edges in $\widetilde{O}(m)$ time.
\vspace{-10pt}

\paragraph{Distributed and Parallel Algorithms for FT-Spanners and Connectivity Certificates.} The main benefit in our algorithm is in its local implementation in bandwidth-restricted distributed models, such as the congest model \cite{Peleg:2000}. The algorithm of Theorem \ref{thm:linear-time-opt-randomized} immediately leads to 
distributed construction with $\widetilde{O}(1)$ congest rounds. Prior algorithms \cite{DR20} obtained FT-spanners with sub-optimal sparsity (by a factor of $f$) within $\widetilde{O}(f^{2-1/k})$ congest rounds. 

\begin{mdframed}[hidealllines=true,backgroundcolor=gray!25]
\vspace{-5pt}
\begin{theorem}\label{thm:linear-time-opt-randomized-distributed}[\~{O}ptimal Spanner in \~{O}ptimal Distributed Time]
There is a randomized distributed congest algorithm that given any $n$-vertex $m$-edge (possibly weighted) graph $G=(V,E,W)$, stretch parameter $k$ and fault bound $f \in [1,n]$ computes in $\widetilde{O}(1)$ rounds nearly optimal $f$-VFT $(2k-1)$ spanner, w.h.p. 
\end{theorem}
\end{mdframed}
A corollary of Theorem \ref{thm:linear-time-opt-randomized-distributed} provides a $\widetilde{O}(1)$-round algorithm for computing $\lambda$-vertex connectivity certificates with nearly optimal bounds. No distributed algorithms have been known to this problem before. For $\lambda$-\emph{edge} connectivity certificates, an $\widetilde{O}(\lambda)$ round algorithm was given by \cite{Par19}.

We then turn to consider the PRAM implementation. Despite its locality, a key computational step in the algorithm of Thm.\ \ref{thm:linear-time-opt-randomized} is in fact sequential\footnote{Of parallel depth $O(\Delta)$, where $\Delta$ is the maximum degree.}. In the distributed setting, this sequential computation is performed locally at each vertex, and therefore does not effect the distributed runtime. In the parallel setting, we overcome this issue but only for unweighted graphs. 
\begin{mdframed}[hidealllines=true,backgroundcolor=gray!25]
\vspace{-5pt}
\begin{theorem}\label{thm:linear-time-opt-randomized-PRAM}[\~{O}ptimal Spanner in \~{O}ptimal Parallel Time]
There is a randomized PRAM (CRCW)  algorithm that given any $n$-vertex $m$-edge \emph{unweighted} graph $G=(V,E)$ computes, w.h.p.:
\begin{enumerate}
\item nearly optimal $f$-VFT $(2k-1)$ spanner using $\widetilde{O}(m)$ work and $\widetilde{O}(1)$ depth. 
\item nearly sparse $\lambda$-vertex connectivity certificates using $\widetilde{O}(m)$ work and $\widetilde{O}(1)$ depth.
\end{enumerate}
\end{theorem}
\end{mdframed}
The prior PRAM algorithms for $f$-VFT $(2k-1)$ spanners, implied by \cite{BaswanaS:07} and \cite{DinitzK:11}, use $\widetilde{O}(f^2 m)$ work and $\widetilde{O}(1)$ depth. These algorithms also work for weighted graphs.

\smallskip

\noindent\textbf{Derandomization.} Our near-optimal algorithms are randomized and provide FT spanners with high probability. We also provide a derandomization of our algorithms at the cost of increasing the runtime by a factor of $f$, in the sequential and the distributed setting. 

\begin{mdframed}[hidealllines=true,backgroundcolor=gray!25]
\vspace{-5pt}
\begin{theorem}\label{thm:suplinear-time-det}[Derandomization]
There are deterministic algorithms that given any $n$-vertex $m$-edge (possibly weighted) graph $G=(V,E,W)$, stretch parameter $k$ and fault bound $f \in [1,n]$ computes optimal FT spanners in $\widetilde{O}(fm)$ sequential time and $\widetilde{O}(f)$ congest rounds. 
\end{theorem}
\end{mdframed}
This improves over the state-of-the-art running time of $O(f^{1-1/k}n^{2+1/k}+f^2 m)$ by \cite{BodwinDR21} (while increasing the spanner size by a poly-log factors). Designing linear-time deterministic algorithms for VFT spanners is one of the interesting open problems. In Our Technique Section, we highlight the critical randomized part in our algorithm. 
\\ 
\noindent \textbf{Discussion and Open Problems.} The Baswana-Sen algorithm is one of the most versatile algorithm known for graph spanners: it has been applied in a diverse range of non-sequential settings, examples include: the distributed congest model \cite{BaswanaS:07,GhaffariK18}, Local Centralized Algorithms (LCA) \cite{ParterRVY19}, the Congested-Clique \cite{Censor-HillelPS17,ParterY18}, the Massively-Parallel-Computation (MPC) models \cite{BiswasDGMN21}. 
It has also been applied in the context of dynamic algorithms for maintaining graph spanners \cite{BaswanaS08,BodwinK16,BernsteinFH19}, dynamic streaming \cite{KapralovW14,FiltserKN21} and more. We therefore believe that its FT extension provided in this paper should have many further, possibly even immediate, implications (e.g., dynamic FT-spanner algorithms).  There are several natural open questions, such as computing truly optimal spanners in linear time, i.e., omitting the extra poly-log factors. We note even the (fault-free) Baswana-Sen algorithm provides $(2k-1)$ spanner with $O(k n^{1+1/k})$ edges. Providing optimal algorithms for optimal EFT-spanners is also a major open problem, in light of the recent work by Bodwin, Dinitz and Robelle \cite{BDRSODA22}.

\vspace{-5pt}\subsection{Our Approach, in a Nutshell}\vspace{-3pt}

We start with a brief overview of the Baswana-Sen algorithm and the highlight the key ideas in our extension to the fault-tolerant setting. To present the key ideas, we assume that the graph $G$ is unweighted, handling weights indeed introduce additional technicalities. 
\vspace{-6pt}
\paragraph{Brief Overview of the Baswana-Sen Algorithm \cite{BaswanaS:07}.} The algorithm has $k$ phases, in which we gradually build (and sometimes dissolve) clusters in an hierarchical manner. A clustering $\mathcal{C}=\{C_1,\ldots, C_\ell\}$ is a collection of vertex-disjoint subsets of vertices. The level-$i$ clustering $\mathcal{C}_i=\{C_{i,1},\ldots, C_{i,\ell}\}$ consists of $\ell=n^{1-i/k}$ clusters, in expectation. Each cluster induces a tree of depth at most $i$ rooted at its cluster \emph{center} and spanning its cluster members. A vertex is denoted as $i$-\emph{clustered} if it belongs to some level-$i$ cluster $C_j \in \mathcal{C}_i$.

The initial clustering is given by $\mathcal{C}_0=\{\{v\}, v \in V\}$, consisting of $n$ singleton clusters. In the $i^{th}$ phase, given is the clustering $\mathcal{C}_{i-1}$ and the $i^{th}$ clustering is defined as follows. Each center of a cluster in $\mathcal{C}_{i-1}$ joins the $i^{th}$ clustering independently with probability of $p=1/n^{1/k}$. A vertex that is adjacent to at least one sampled cluster (i.e., with a sampled center) joins that cluster (by connecting to its neighbor in that cluster). All other $(i-1)$-clustered vertices, with no adjacent sampled clusters, add to the spanner, one edge to some of their neighbors for each of their adjacent clusters in $\mathcal{C}_{i-1}$. The clustering $\mathcal{C}_k$ is defined to be the empty, therefore all vertices are $k$-unclustered. 

Letting $H$ denote the output spanner, the stretch argument shows that for every $i$-unclustered vertex $v$, it holds that $\dist_H(u,v)\leq 2i-1$ for every $u \in N(v)$, where $N(v)$ are the neighbors of $v$ in $G$. The size analysis is based on showing that an $i$-unclustered vertex is incident to $O(n^{1/k})$ clusters in $\mathcal{C}_{i-1}$, in expectation. This construction is clearly not resilient, not even against a single vertex fault. For example, the stretch argument might be broken in the case where the unique path connecting a vertex to its cluster center contains a faulty vertex. 
\vspace{-6pt}
\paragraph{New Idea: Fault-Tolerant Clustering.} Our extension of the Baswana-Sen algorithm is also based on $k$ levels of clustering, where level-$i$ clusters induce depth-$i$ trees in $G$. Our clusters, however, have additional FT properties: the level-$i$ clusters of $\mathcal{C}_{i}$ are not vertex-disjoint, but rather induce vertex-independent trees in $G$ (i.e., the tree paths from the roots to some fixed vertex $v$ are vertex-disjoint).  
Importantly, each $i$-clustered vertex is required to appear in $\Omega(f)$ of these trees. This guarantees that for any sequence of at most $f$ vertex faults $F$, every $i$-clustered vertex $v$ has at least one \emph{fault-free} path (i.e., that does not intersect $F$), of length at most $i$, to one of its cluster centers in $\mathcal{C}_{i}$. Our main challenge is in defining these clusters, while guaranteeing that the edges incident to the $i$-unclustered vertices are already $(f,i)$ protected, as defined next.

\begin{definition}[$(f,i)$-protection]\label{def:protect}
An edge $e=(u,v)$ is $(f,i)$-protected in a subgraph $H' \subseteq G$ if 
$$\dist_{H' \setminus F}(u,v)\leq (2i-1)\cdot W(e), \forall F \subseteq V \setminus\{u,v\}, |F|\leq f~.$$
A subgraph $H$ is then $f$-VFT $(2k-1)$ \emph{spanner} iff every edge $e \in G$ is $(f,k)$ protected in $H$. 
\end{definition}
 
\noindent \textbf{Building the $i^{th}$ Level of the FT-Clustering. } We highlight some of the key ideas in computing the $i^{th}$ level $\mathcal{C}_i$ given the $(i-1)^{th}$ level clustering $\mathcal{C}_{i-1}$.

Consider the set $V_{i-1}$ of $(i-1)$-clustered vertices (under some delicate definition that we skip for now\footnote{there are vertices in $\mathcal{C}_{i-1}$ that are in fact $(i-1)$-\emph{unclustered}}). Focus on some vertex $v \in V_{i-1}$. A collection of clusters $\mathcal{C}_v \subseteq \mathcal{C}_{i-1}$ is denoted as an \emph{independent set of adjacent clusters} of $v$ if: (i) each cluster $C \in \mathcal{C}_v$ contains some neighbor of $v$, and (ii) the collection of paths connecting $v$ to the centers of the clusters of $\mathcal{C}_v$ (via $v$'s neighbors in these clusters) are \emph{vertex-disjoint}, except for the common endpoint $v$. In the $i^{th}$ step of the algorithm, it is first desired to compute a maximum set of independent adjacent clusters for $v$ in $\mathcal{C}_{i-1}$. As this might hard, we settle for a \emph{maximal} independent set, up to an additional \emph{slack} that is introduced due to the running time considerations. Note that as each vertex $u$ might belong to $\Omega(f)$ clusters in $\mathcal{C}_{i-1}$, a vertex $v$ might be incident to $\Omega(f deg(v))$ clusters in $\mathcal{C}_{i-1}$. Therefore, a na\"{\i}ve computation of a maximal independent set of adjacent clusters might take $\Omega(f deg(v))$ time per vertex $v$, and consequently $\Omega(f m)$ time for all vertices. 

This is the critical point where randomization comes to rescue! Instead of considering the $\Omega(f)$ clusters of each of neighbor $u$ of $v$, we a-priori sub-sample for each vertex $u$, a collection of $O(\log n)$ clusters sampled uniformly at random from the $\Omega(f)$ clusters to which it belongs. The sampling step is done only once and therefore takes $O(f n)$ time.
When computing the MIS $\mathcal{C}_v$ for $v$, we restrict attention only to the union of $O(deg(v)\log n)$ adjacent sampled clusters. The set $\mathcal{C}_v$ is then obtained by applying a simple greedy MIS procedure to find the maximal independent set of adjacent clusters (i.e., add cluster $C$ is added $\mathcal{C}_v$, only if the tree path from $v$ to the center of $C$ is vertex-disjoint w.r.t clusters already taken into $\mathcal{C}_v$). 

The $i^{th}$-level \emph{centers} are defined by sampling each $(i-1)^{th}$-level center independently with probability of $p=(f/n)^{1/k}$. Therefore, in expectation, we have $O(f^{i/k}n^{1-i/k})$ clusters in level $i$. A vertex $v$ will be defined as $i$-\emph{clustered} only if its MIS set $\mathcal{C}_v$ hits at least $\Theta(kf)$ sampled clusters\footnote{The dependency in $k$ becomes clear in the analysis part.}. The $i$-clustered vertices will be then safely added to some $\Theta(kf)$ sampled adjacent clusters in $\mathcal{C}_v$. This will preserve the vertex-independency of the trees induced by the 
$i^{th}$-level clustering. The main challenge is in handling the $i$-\emph{unclustered} vertices. Here,  we need to account for the fact that we computed the MIS 
$\mathcal{C}_v$ \emph{only} based on a random sample of adjacent clusters (this is quite challenging especially in the weighted setting where our uniform sampling of adjacent clusters totally ignore weights, and consequently might miss the nearest clusters). Our stretch analysis is based on a key structural lemma (Lemma \ref{lem:path-pairs}), which constitutes our main technical contribution. On a high level, the stretch argument is based on showing that the $v$-edges $(u,v)$ incident to a cluster not taken into the MIS $\mathcal{C}_v$ are w.h.p. either (i) taken into the spanner, (ii) admit a sufficient number of $u$-$v$ near-vertex disjoint paths in the current spanner, or else (iii) be handled in the subsequent phases (in the latter case, both $u$ and $v$ are $i$-clustered).

To parallelize this construction, we need to provide an efficient algorithm for computing the maximal independent set $\mathcal{C}_v$. Our above mentioned procedure works in a sequential greedy manner, hence requiring a parallel depth of $\Theta(\Delta)$, where $\Delta$ is the maximum degree.  Our solution is based on the parallel greedy MIS algorithm of Blelloch, Fineman and Shun \cite{BlellochFS12}. The efficient implementation of this algorithm exploits the fact that each \emph{node} in the conflict graph in our MIS instance corresponds to a short \emph{path}. The PRAM computation computes a random permutation $\pi$ over the $O(deg(v)\log n)$ potential paths and then compute the lexicographic-first MIS collection of paths with respect to $\pi$ using $\widetilde{O}(deg(v))$ work and $\widetilde{O}(1)$ depth, per vertex $v$. 

\smallskip

\noindent\textbf{Preliminaries.}
Throughout, we consider an $n$-vertex $m$-edge undirected weighted graph $G=(V,E,W)$. We assume that the edge weights are polynomial in $n$\footnote{This assumption can avoided at the cost of increasing the runtime by a factor of $O(\log W_{\max})$.} and unique (this can be obtained by appending the unique edge ID to break the tie). In the case where the graph is unweighted, all edges have weight of $1$. For edge pair $e,e'$, we use $e<e'$ to indicate that $W(e)<W(e')$. For $v \in V$, let $N(v,G)$ be the neighbors of $v$ in $G$, and $\deg(v,G)=|N(v,G)|$, when $G$ is clear from the context, we may omit it.  
For a path $P=[u_0,\ldots, u_k=v]$ ending at $v$, we refer to the \emph{first} vertex on the path by $h(P)$ (the \emph{head} of the path $P$). In the same manner, let $t(P)=u_k$ be the \emph{tail} of the path. For $u_{j}, u_{j'}$ where $j \leq j'$, denote the $u_{j}$-$u_{j'}$ subpath in $P$ by $P[u_{j}, u_{j'}]$. In case where $u_j=h(P)$, we simply write $P[\cdot, u_j]$ (instead of $P[u_0,u_j]$). Let $|P|$ denote the number of edges (hops) in $P$, and let $len(P)=\sum_{e \in P}W(e)$, i.e., the \emph{length} of $P$. The last edge of the path $P$ is denoted by $LastE(P)$. For paths $P_1,P_2$ with $t(P_1)=h(P_2)$, the concatenation of these paths is denoted by $P_1 \circ P_2$.
Let $\dist_G(u,v)$ denote length of shortest $u$-$v$ path in $G$ (i.e., weighted distance). 
For a collection of paths $\mathcal{P}$, the vertices of $\mathcal{P}$ is denoted by $V(\mathcal{P})=\bigcup_{P \in \mathcal{P}}V(P)$. For a set of elements $X$ and $p \in (0,1)$, let $X[p]$ be a subset of $X$ obtained by sampling each element of $X$ independently with probability $p$.

\smallskip 
\noindent \textbf{Roadmap.} For a warm-up, we start with proving Theorem \ref{thm:linear-time-opt-randomized} for $k=2$. Already this basic setting calls for some new ideas. In Sec. \ref{sec:spanners-general-meta}, we describe the meta-algorithm for computing nearly optimal FT-spanners for any stretch $k$. In Sec. \ref{sec:spanner-imp}, we provide the implementation details in the sequential and distributed models. Then in Sec. \ref{sec:parallel-imp}, we show the parallel implementation, which requires additional ideas. Sec. \ref{sec:certificate} considers vertex-certificates, proving the proof of Theorem \ref{thm:linear-time-opt-randomized-PRAM}(2). Finally, in Sec. \ref{sec:derand} we derandomize the constructions.



\vspace{-10pt}
\section{Warm Up: $f$-VFT $3$-Spanners} \vspace{-5pt}
As a warm-up, we provide a simplified variant of our algorithm for $3$-spanners (i.e., $k=2$). We present a sequential spanner construction with $\widetilde{O}(\sqrt{f}n^{3/2})$ edges, that runs in $\widetilde{O}(m)$ time.
The algorithm has two main steps. The first defines a collection of $\ell=O(\sqrt{f n})$ (in expectation) star clusters $\mathcal{C}=\{C_1, \ldots, C_\ell\}$. A vertex $v$ will be defined as \emph{clustered} if it belongs to $4f$ clusters, and otherwise, it will be defined as \emph{unclustered}. All edges adjacent to unclustered vertices will be added to the spanner $H$. The second phase carefully connects each clustered vertex to each of its adjacent clusters in $\mathcal{C}$. 

\smallskip

\noindent\textbf{Step One: Clustering.} Let $S=V[p]$ be a subset of vertices, denoted as \emph{centers} where $p=\sqrt{f/n}$. For each vertex $v$, let $E(v)=\{(u_1,v), \ldots, (u_q,v)\}$ be the edges adjacent to $v$ sorted in increasing edge weight, i.e., $(u_1,v)<(u_2,v)<\ldots<(u_q,v)$. If $v$ is adjacent to at least $4f$ vertices in $S$ (i.e., $|N(v)\cap S|\geq 4f$), it is denoted as \emph{clustered} and otherwise it is \emph{unclustered}. The subset of unclustered vertices is denoted by $U$. For every clustered vertex $v$, let $S(v) \subseteq S \cap N(v)$ be the $4f$ \emph{nearest} sampled neighbors of $v$ based on the weights of $W$. For unweighted graphs, the set of $4f$ centers in $S(v)$ can be chosen in an arbitrary manner from the set $N(v)\cap S$. We are now ready to define the clustering. 
For every $s \in S$, let $C(s)=\{ v ~\mid~ s \in S(v)\}$ and define the clustering $\mathcal{C}=\{C(s) ~\mid~ s \in S\}$. 
Note that a clustered vertex belongs to $4f$ clusters in $\mathcal{C}$. 

For a clustered vertex $v$, let $LE(v)=\{ (u,v) ~\mid~ W((u,v))< \max_{s \in S(v)}W((s,v))\}$ be the set of edges adjacent to $v$ that are lighter than the maximum edge weight connecting $v$ to its centers in $S(v)$. For an unclustered vertex $v$, let $LE(v)=E(v)$. At the end of the step, the algorithm defines:
$$H'= \{(s,v) ~\mid~ s \in S, v \in C(s)\} \cup \bigcup_{v \in V}LE(v)~.$$

\noindent \textbf{Step Two: Maximal Independent Set of Adjacent Clusters.} For a clustered vertex $v$, define $E'(v)=E(v)\setminus H'$. In addition, let $S'(v)\subseteq S(v)$ be a subset of $\ell=O(\log n)$ centers, obtained by having $\ell$ independent random uniform samples from $S(v)$.  For every clustered vertex $v$, the algorithm iterates over the edges in $E'(v)$ in increasing edge weights (from the lightest to the heaviest). It maintains a list $L(v)$ of the set of observed centers, and an edge set $\widetilde{E}(v)$ defined as follows. Initially, let $L(v),\widetilde{E}(v)=\emptyset$. Then, when considering the $j^{th}$ edge $(u_j,v)$ in $E'(v)$, the algorithm adds $(u_j,v)$ to $\widetilde{E}(v)$ only if 
$$S'(u_j) \setminus L(v) \neq \emptyset~.$$
In the latter case, it also adds one arbitrary cluster center in $S'(u_j) \setminus L(v)$ to the list $L(v)$.  
Finally, let $H=H' \cup \bigcup_{v \in V \setminus U}\widetilde{E}(v)$. This completes the description of the algorithm. 

\smallskip

\noindent \textbf{Size Analysis.} Using standard application of the Chernoff bound, we have that w.h.p. $|LE(v)|=O(f\log n/p)=O(\sqrt{f n}\log n)$. Hence, $|H'|=O(fn+\sqrt{f}n^{3/2}\log n)$. In the second step, each vertex adds at most one edge for each adjacent cluster. By Chernoff again, w.h.p. the number of clusters is $|\mathcal{C}|=O(\sqrt{f n}\log n)$, therefore, we add $O(\sqrt{f}n^{3/2}\log n)$ edges, in total.

\smallskip 
\noindent \textbf{Stretch Analysis.} Since all the edges adjacent to the unclustered vertices $U$ are in $H$, it is sufficient to consider an edge $(u,v)$ where $u$ and $v$ are clustered, and that in addition, $(u,v)\notin LE(v) \cup LE(u)$.
Let $L(v,u)$ be the set of centers in $L(v)$ just before considering the edge $(u,v)$ (at the point of considering the vertex $v$). Recall, that the edge $(u,v)$ is added to the spanner only if $S'(u)\setminus L(v,u)\neq \emptyset$. We will show that w.h.p.\ either $(u,v)\in  H$, or else, $H$ contains at least $(f+1)$ vertex-disjoint $u$-$v$ paths, each of length at most $3W((u,v))$. We distinguish between two cases. 
\\
\noindent \textbf{Case 1: $|L(v,u) \cap S(u)|\leq |S(u)|/2$.} We claim that, w.h.p., in this case $S'(u)\setminus L(v,u)\neq \emptyset$, and therefore $(u,v) \in H$. To see this observe that by sampling a single vertex uniformly at random from $S(u)$, with probability of $1/2$ this sampled vertex is \emph{not} in $L(v,u)$. Therefore, w.h.p., by making $O(\log n)$ samples (i.e., the set $S'(u)$), at least one of the sampled center is not $L(v,u)$, as desired.
\\ \\
\noindent \textbf{Case 2: $|L(v,u) \cap S(u)|> |S(u)|/2$.} Let $\widetilde{E}(v,u)$ be the edge set of $v$ just before considering the edge $(u,v) \in E'(v)$. For each edge $(z,v) \in \widetilde{E}(v,u)$, let $s(z)$ be the unique center in $S'(z)$ added to $L(v,u)$. Then, in this case $\widetilde{E}(v,u)$ contains at least $|S(u)|/2 \geq 2f$ edges $(u'_1,v),\ldots, (u'_q,v)$ that are all lighter than $(u,v)$ (by the ordering of $E'(v)$), and in addition, $s(u'_1),\ldots, s(u'_q) \in S(u)$. Since the algorithm adds to $H$ all edges in $\widetilde{E}(v,u)$, it implies that $H$ contains $2f$ vertex-disjoint $u$-$v$ paths, each with $3$ edges, given by
$$Q_j=(v,u'_j) \circ (u'_j,s(u'_j)) \circ (s(u'_j),u) \mbox{~for~} j \in \{1,\ldots, 2f\}~.$$

We next claim that $(u,v)$ is heavier than all edges in $Q_j$ for every $j$. By the ordering of edges in $E'(v)$, we have $(u,v)>(u'_j,v)$. Since $(v,u'_j) \notin H'$, it implies that $(v,u'_j) \notin LE(u'_j)$, and therefore $(v,u'_j)>(u'_j,s(u'_j))$ and combining with the above $(u,v)>(u'_j,s(u'_j))$. In the same manner, since $(u,v) \notin LE(u)$, we have that 
$(u,v)>(u,s(u'_j))$. Overall, w.h.p., $H$ consists of $2f$ vertex-disjoint $u$-$v$ paths, each of length at most $3W(e)$. Therefore, for any $F \subseteq V \setminus \{u,v\}$ where $|F|\leq f$, there exists some $u$-$v$ path $Q_j \subseteq H \setminus F$, concluding that $\dist_{H \setminus F}(u,v)\leq 3W(e)$. 

\smallskip

\noindent \textbf{Sequential Running time.} We can assume that $m=\Omega(\sqrt{f}\cdot n^{3/2})$, as otherwise, $H=G$. Sorting the edges $E(v)$ for every $v$, and computation of the clusters can be done in $\widetilde{O}(m)$ time. Computing the set $LE(v)$ can be done in $O(m)$ time. The computation of the sets $S'(v) \subseteq S(v)$ takes $\widetilde{O}(f n)$ time. Fix a vertex $v$. We store the list $L(v)$ using a hash table, and thus computing $\widetilde{E}(v)$ takes $O(deg(v)\log n)$ time, and $\widetilde{O}(m)$ in total. Observe that computation of the sets $S'(u)$ is crucial to provide the $\widetilde{O}(m)$ runtime, otherwise computing $L(v)$ might take $O(deg(v)f)$ time, and $\widetilde{O}(f m)$ in total.

\vspace{-7pt}\section{$f$-FT Vertex $(2k-1)$ Spanners}\label{sec:spanners-general-meta}
In this section, we prove Theorem \ref{thm:linear-time-opt-randomized} and \ref{thm:linear-time-opt-randomized-distributed}. We first describe the meta-algorithm, provide its analysis and then in Subsec. \ref{sec:spanner-imp} present the sequential and distributed implementation details.
%
%

\vspace{-7pt}
\subsection{Description of the Meta-Algorithm}
Similarly to the Baswana-Sen algorithm, the FT-spanner algorithm has $k$ phases, where in each phase $i \in \{1,\ldots, k\}$ given the current clustering $\mathcal{C}_{i-1}$, the algorithm computes a new clustering $\mathcal{C}_i$ and in addition, adds edges to the spanner to handle the newly unclustered vertices. In contrast to the standard Baswana-Sen algorithm, however, the properties of the clusters are quite different. For example, the trees induced by the clusters of $\mathcal{C}_i$ are not vertex-disjoint, but rather \emph{vertex-independent}. 
\begin{definition}[Vertex-Independent Trees]\label{def:indep}
A collection of (not necessarily spanning) trees $T_1,\ldots, T_\ell \subseteq G$ are \emph{vertex-independent} if the following holds for every vertex $v \in V(G)$: the collection of root to $v$ paths in the trees containing $v$ are vertex-disjoint (except for their common endpoint $v$).
\end{definition} 
A crucial requirement for FT clustering is to include each clustered vertex in $\Omega(f)$ clusters. Starting with trivial singleton clustering $\mathcal{C}_0=\{\{v_1\}, \ldots, \{v_n\}\}$, the algorithm maintains the following invariants for the clustering $\mathcal{C}_i$ for every $i \in \{0, \ldots, k-1\}$:
\begin{itemize}
\item[I.]   $|\mathcal{C}_i|=O(f^{i/k}\cdot n^{1-i/k})$.
\item[II.]  Each cluster $C_{i,j} \in \mathcal{C}_i$ has a designated center $s_{i,j}$, and an $i$-depth tree $T_i(s_{i,j})$ rooted at $s_{i,j}$ that spans its cluster vertices. Every cluster has a distinct center. 
\item[III.] The collection of $\{T_{i}(s_{i,j})\}$ trees are vertex-independent.
\item[IV.] The edge weights along every root to leaf $v$ path in $T_{i,j}$ are monotone increasing\footnote{Towards the leafs.}.
\item[V.] Each $i$-clustered vertex for $i\geq 1$, belongs to $K_f=20kf$ clusters in $\mathcal{C}_i$. 
\end{itemize}
Properties (I-IV) hold vacuously for $\mathcal{C}_0$. We need the following definitions. A path $P$ \emph{intersects} a path collection $\mathcal{P}$ if there exists at least one path $P' \in \mathcal{P}$ such that $V(P)\cap V(P')\neq \emptyset$. We write $e > \mathcal{P}$ if $W(e)> W(e')$ for every $e' \in \bigcup_{P \in \mathcal{P}} E(P)$. We are now ready to describe the $i^{th}$ phase of the algorithm. 

\vspace{-5pt}
\subsubsection*{Description of Phase $i$} 
The input to this phase is as follows:
\begin{itemize}
\item the $(i-1)^{th}$ clustering $\mathcal{C}_{i-1}$ (satisfying invariants (I-IV)),
\item the set $V_{i-1} \subseteq V$ of $(i-1)$-clustered vertices, 
\item a collection of tree paths $\mathcal{Q}_{i-1}(v)$ that connects $v$ to its clusters in 
$\mathcal{C}_{i-1}$, for every $v \in V_{i-1}$, 
\item a subset $R_{i-1} \subseteq E$ of remaining edges to protect.  
\end{itemize}
Initially, $\mathcal{C}_0=\{\{v_1\}, \ldots, \{v_n\}\}$, $V_0=V$, $\mathcal{Q}_{0}(v)=\{\{v\}\}$ and $R_0=E(G)$. 
\vspace{-5pt}
\paragraph{Step 1: Maximal Independent Sets of Adjacent Clusters.} The goal of this step is to compute for each vertex $v \in V_{i-1}$ a maximal collection of adjacent clusters in $\mathcal{C}_{i-1}$. This path collection will be denoted by $\mathcal{P}^*_{i-1}(v)$, and will later on determine the clustering $\mathcal{C}_{i}$. The path collection $\mathcal{P}^*_{i-1}(v)$ is defined in two steps. The first major step computes a preliminary maximal independent set of paths $\mathcal{P}_{i-1}(v)$ to the centers in $Z_{i-1}$. For the purpose of handling edge weights, the algorithm employs a cleanup procedure on $\mathcal{P}_{i-1}(v)$ that yields the final set $\mathcal{P}^*_{i-1}(v)$.

By properties (III,V) for $\mathcal{C}_{i-1}$, we have that for every $u \in V_{i-1}$, $\mathcal{Q}_{i-1}(u)$ consists of $K_f=20kf$ vertex-disjoint paths (except for common endpoint $u$).
The step starts by sampling for every $u \in V_{i-1}$ a collection $\mathcal{S}_{i-1}(u) \subset \mathcal{Q}_{i-1}(u)$ of $O(\log n)$ paths sampled independently and uniformly at random from $\mathcal{Q}_{i-1}(u)$. 
For our purposes, it is sufficient to perform the sampling of $\mathcal{S}_{i-1}(u)$ only once, and these sampled sets will be then used in defining $\mathcal{P}^*_{i-1}(v)$ for every $(i-1)$-clustered neighbor $v$ of $u$. The only purpose for defining $\mathcal{S}_{i-1}(u)$ (rather than simply working with the sets 
$\mathcal{Q}_{i-1}(u)$) is for the sake of optimizing the running time\footnote{Using the sets $\{\mathcal{Q}_{i-1}(u)\}$ instead of the sampled sets $\{\mathcal{S}_{i-1}(u)\}$ leads to a running time of $\widetilde{O}(f m)$, which is too costly for proving Theorem \ref{thm:linear-time-opt-randomized}.}.

We next focus on a vertex $v \in V_{i-1}$ and explain how to compute the MIS of paths $\mathcal{P}_{i-1}(v)$. Initially, $\mathcal{P}_{i-1}(v)=\mathcal{Q}_{i-1}(v)$. Let $E_{i-1}(v)=\{(u_1,v), (u_2,v), \ldots, (u_\ell,v)\} \subseteq R_{i-1}$ be $v$'s edges in $R_{i-1}$ sorted in increasing edge weights. The edges of $E_{i-1}(v)$ are traversed one by one, where in iteration $j$ the algorithm considers the edge $(u_j,v)$, and adds at most one path to $\mathcal{P}_{i-1}(v)$, as follows. 

\smallskip 
\noindent\textbf{The $j^{th}$ Iteration:} 
If there exists a path in $\mathcal{S}_{i-1}(u_j)$ that does not intersect the current set of paths $\mathcal{P}_{i-1}(v)$, the algorithm picks one such a path, arbitrarily, denoted as $P^*_{u_j,v} \in \mathcal{S}_{i-1}(u_j)$. The algorithm then adds the path $P=P^*_{u_j,v}\circ (u_j,v)$ to the path collection $\mathcal{P}_{i-1}(v)$. This completes the description of the $j^{th}$ iteration. Note that we indeed preserve the property that all paths in $\mathcal{P}_{i-1}(v)$ are vertex-disjoint (except for the common $v$).  See Fig. \ref{fig:MIS} for an illustration.
\\
\\
\noindent\textbf{Cleanup Step: Shortcutting $\mathcal{P}_{i-1}(v)$.}  We next perform some cosmetic modifications for the paths in $\mathcal{P}_{i-1}(v)$ which are needed, in our analysis, only for the weighted case (that is, this step can be skipped for unweighted graphs). Observe by the end of traversing all edges in $E_{i-1}(v)$, the paths in $\mathcal{P}_{i-1}(v)$ are vertex-disjoint (except for the common endpoint), and start at unique centers in $Z_{i-1}$. The algorithm applies a shortcut procedure $\Shortcut(P)$ on each $P \in \mathcal{P}_{i-1}(v)$. For every path $P=[u_0,\ldots, u_{q}=v] \in \mathcal{P}_{i-1}(v)$, consider the edges $\{ (u, v) ~\mid~ (u,v) \in R_{i-1}, u \in V(P)\}$, and let $(u_j,v)$ be the edge of minimum weight in this list. Define $\Shortcut(P)=P[u_0,u_j]\circ (u_j,v)$ (see Fig. \ref{fig:multi-paths} (Right)), and let
\begin{equation}\label{eq:shortcut-paths}
\mathcal{P}^{*}_{i-1}(v)=\mathcal{Q}_{i-1}(v) \cup \{\Shortcut(P)~\mid~ P \in \mathcal{P}_{i-1}(v)\setminus \mathcal{Q}_{i-1}(v)\}~.
\end{equation}
Note that the paths of $\mathcal{P}^{*}_{i-1}(v)$ are still vertex-disjoint (except for the common endpoint, $v$).
For ease of notation we need the following definition, for a path $P' \in \mathcal{P}^{*}_{i-1}(v)$, let $\Shortcut^{-1}(P')$ be the path $P \in \mathcal{P}_{i-1}(v)$ satisfying that $\Shortcut(P)=P'$. That is, there is a bijection from the paths in $\mathcal{P}_{i-1}(v)$ to the paths in $\mathcal{P}^*_{i-1}(v)$. 

\vspace{-5pt}
\paragraph{Step 2: Defining the $i^{th}$ Clustering $\mathcal{C}_i$.}
Let $Z_{i}=Z_{i-1}[p]$ for $p=(f/n)^{1/k}$ for every $i \in \{1,\ldots, k-1\}$. For $i=k$, let $Z_k=\emptyset$.
A vertex $v$ is denoted as $i$-\emph{clustered} if $\mathcal{P}^*_{i-1}(v)$ \emph{hits} at least $K_f=20kf$ sampled centers. Formally, letting $Z'_{i-1}(v)=\{h(P)~\mid~ P \in \mathcal{P}^*_{i-1}(v)\}$, then $v$ is $i$-\emph{clustered} if $|Z'_{i-1}(v) \cap Z_{i}|\geq K_f$, and otherwise it is $i$-\emph{unclustered}. Note that all the vertices are $k$-unclustered.

For every $i$-clustered $v$, the algorithm adds $v$ to $V_i$ and connects it to the clusters of $K_f$ sampled centers (in $Z_i$) carefully chosen, as follows. Let $\mathcal{P}^*_{i-1}(v)=\{P_1,\ldots, P_\ell\}$ be ordered in increasing edge weights of their last edge, i.e., $LastE(P_1)<LastE(P_2)< \ldots< LastE(P_\ell)$.
We say that a path $P \in \mathcal{P}^*_{i-1}(v)$ is \emph{sampled} if its first vertex is sampled into $Z_i$, i.e., $h(P) \in Z_{i}$. The algorithm then selects the first $K_f$ sampled paths $P_{j_1},\ldots, P_{j_{K_f}}$ in the ordered set $\mathcal{P}^*_{i-1}(v)$. Let 
$$\mathcal{Q}_i(v)=\{P_{j_1},\ldots, P_{j_K}\}, j_1< j_2 < \ldots < j_{K_f}~.$$  
For every $s \in Z_i$, define its spanning tree $T_i(s)$ and its cluster $C_i(s)$ by
\begin{equation}\label{eq:tree}
T_i(s)=\{ P ~\mid~ P \in \mathcal{Q}_i(v), ~h(P)=s, ~v \in V_i\}~,~ C_i(s)=V(T_i(s))~.
\end{equation}
In the analysis, we show that each $T_i(s)$ is indeed a tree of depth at most $i$. 
The $i^{th}$ clustering is given by $\mathcal{C}_i=\{C_i(s) ~\mid~ s \in Z_i\}$. 
\vspace{-5pt} 

\paragraph{Step 3: Defining the Spanner $H_i$ and the Remaining Edge Set $R_i$.}
For an $i$-clustered vertex $v \in V_i$ with $\mathcal{P}^*_{i-1}(v)=\{P_1,\ldots, P_\ell\}$, let $i_v$ be the largest path index taken into $\mathcal{Q}_i(v)$. For an $i$-unclustered vertex $v$, let $i_v=|\mathcal{P}^*_{i-1}(v)|+1$. Define the set of $v$ edges:
\begin{equation}\label{eq:LEiv}
LE_i(v)=\bigcup_{j \leq i_v-1} \{ (u,v) ~\mid~ (u,v) \in R_{i-1}, u \in V(\Shortcut^{-1}(P_j))\}~.
\end{equation}
Note that in the above we consider any $u \in V(\Shortcut^{-1}(P_j))$ and that $V(P_j)\subseteq V(\Shortcut^{-1}(P_j))$. This will be important for the stretch analysis. 
The output spanner $H_i$ is then given by
\begin{equation}\label{eq:Hi}
H_i=\bigcup_{s \in Z_i} E(T_i(s)) \cup \bigcup_{v \in V_{i-1}} LE_i(v) \cup H_{i-1}~.
\end{equation}
The remaining edge set to be handled is defined by  
\begin{equation}\label{eq:Riv}
R_{i}=\{(u,v) \in V_{i} \times V_{i} ~\mid~ (u,v) \notin H_i, (u,v)>\mathcal{Q}_{i}(u) \cup \mathcal{Q}_{i}(v)\}~.
\end{equation}
This completes the description of phase $i \in \{1,\ldots, k\}$. The output spanner is given by $H=H_k$. A succinct pseudocode of phase $i$ is given below. 

\begin{mdframed}[hidealllines=false,backgroundcolor=gray!00]
\center \textbf{Phase $i$ of Alg. $\VertexFTSpanner$:}
\begin{flushleft}
\textbf{Input:} Clustering $\mathcal{C}_{i-1}$, subset $V_{i-1}$ of $(i-1)$ clustered vertices, remaining edge set $R_{i-1}$.
\textbf{Output:} Clustering $\mathcal{C}_{i}$, subset $V_{i}\subseteq V_{i-1}$, $H_{i}$ and $R_i \subseteq R_{i-1}$.
\end{flushleft}

\vspace{-8pt}
\begin{itemize}
\item{\textbf{Step 1: Maximal Disjoint Paths to Adjacent Clusters in} $\mathcal{C}_{i-1}$.}
\begin{enumerate}
\item For each $v \in V_{i-1}$, let $\mathcal{S}_{i-1}(v)$ be a random sample of $O(\log n)$ paths from $\mathcal{Q}_{i-1}(v)$. 

\item For every $v \in V_{i-1}$ compute a set $\mathcal{P}^{*}_{i-1}(v)$ as follows:
\begin{enumerate}
\item Set $\mathcal{P}_{i-1}(v)\gets \mathcal{Q}_{i-1}(v)$.
\item Set $E_{i-1}(v)=\{(u_1,v),\ldots, (u_\ell,v)\} \subseteq R_{i-1}$ ordered in increasing edge weights. 
\item{\textbf{Iteration} $j \in \{1,\ldots \ell\}$:} 
\begin{itemize}
\item If exists $P \in  \mathcal{S}_{i-1}(u_j)$ such that $V(P) \cap V(\mathcal{P}_{i-1}(v))=\emptyset$:
\begin{itemize} \item add $P \circ (u_j,v)$ to $\mathcal{P}_{i-1}(v)$. 
\end{itemize}
\end{itemize}
\item $\mathcal{P}^{*}_{i-1}(v)=\{\Shortcut(P) ~\mid~ P \in \mathcal{P}_{i-1}(v)\}$.
\end{enumerate}
\end{enumerate}

\item{\textbf{Step 2: Computing the} $i^{th} $\textbf{Clustering} $\mathcal{C}_{i}$.}
\begin{enumerate}
\item For $i \leq k-1$: $Z_i=Z_{i-1}[p]$ for $\mathbf{p=(f/n)^{1/k}}$.
\item For $i=k$: Let $Z_k=\emptyset$. 
\item Let $\mathcal{Q}_{i}(v)$ be the $\mathbf{K_f=20kf}$ nearest sampled paths in $\mathcal{P}^{*}_{i-1}(v)$ (if exists).
\item Let $V_i=\{v \in V_{i-1} ~\mid~ |\mathcal{Q}_{i}(v)|=K_f\}$.
\item Let $C_i(s)=\bigcup_{v \in V_{i-1}, P \in \mathcal{Q}_{i}(v), h(P)=s}V(P)$ and $\mathcal{C}_i=\{C_i(s) ~\mid~ s \in Z_i\}~.$
\end{enumerate}

\item{\textbf{Step 3: Defining} $H_i$ \textbf{and} $R_i$ (See Eq. (\ref{eq:Hi},\ref{eq:Riv}))}


\end{itemize}
\end{mdframed}

\begin{figure}[h!]
\begin{center}
\includegraphics[scale=0.35]{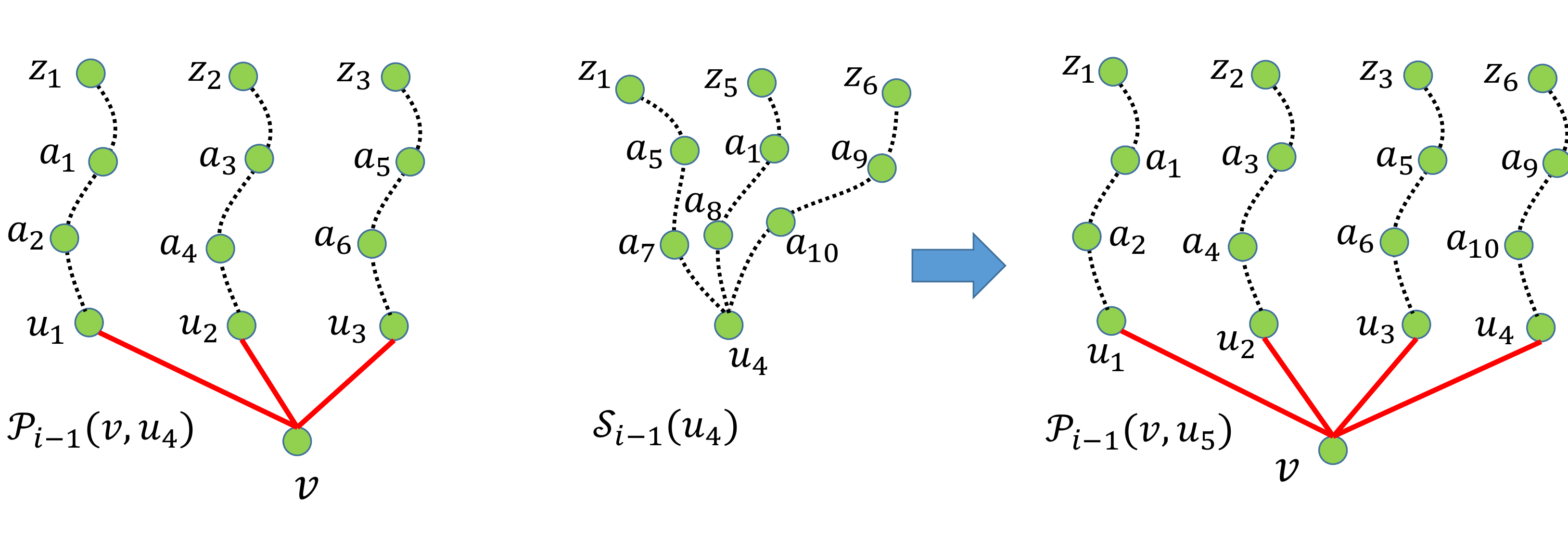}
\caption{\sf An illustration for Step (1) of Alg. $\VertexFTSpanner$. Focus on $v \in V_{i-1}$. The set $\mathcal{P}_{i-1}(v,u_4)$ is the collection of the vertex-disjoint paths (except for $v$) given at the beginning of the $4^{th}$ iteration. In that iteration, the algorithm considers the edge $(v,u_4)$ and scans all paths in $\mathcal{S}(u_4)$ in order to detect a path that is vertex-disjoint from all current paths in $\mathcal{P}_{i-1}(v,u_4)$. In the example, the path $[z_6,a_9,a_{10},u_4]$ is added to the output set $\mathcal{P}_{i-1}(v,u_5)$ (i.e., the input for the $5^{th}$ iteration). \label{fig:MIS} 
}
\end{center}
\end{figure}
\vspace{-3pt}
\subsection{Size Analysis and Auxiliary Claims} \label{sec:size-claims}
Missing proofs in this section are deferred to Appendix \ref{sec:missproof}. 

\begin{observation}\label{obs:cluster-cont}
For every $P \in \mathcal{Q}_{i-1}(v)$ such that $h(P)\in Z_i$, it also holds that $P \in \mathcal{Q}_{i}(v)$
\end{observation}
\def\APPENDOBSCLUSTCONT{
\begin{proof}[Proof of Obs. \ref{obs:cluster-cont}]
By definition $\mathcal{Q}_{i-1}(v) \subseteq \mathcal{P}^*_{i-1}(v)$. 
Since the last edges of the shortcut paths 
$$\{\Shortcut(P) ~\mid~ P \in \mathcal{P}_{i-1}(v) \setminus \mathcal{Q}_{i-1}(v)\}$$ are also in $R_{i-1}$ and incident to $v$, we have each of their last edges is heavier than the last edges of $\mathcal{Q}_i(v)$. Therefore the paths of $\mathcal{Q}_{i-1}(v)$ appear \emph{first} in the ordered set $\mathcal{P}^*_{i-1}(v)$. We have that for each sampled center $h(P) \in Z_i$ for every $P \in \mathcal{Q}_{i-1}(v)$, it holds that $P \in \mathcal{Q}_{i}(v)$. 
\end{proof}
}

\begin{observation}\label{obs:cluster-induc}
For every $v \in V_i$ and a path $P \in \mathcal{Q}_i(v)$, it holds that $P[\cdot, u] \in \mathcal{Q}_{i-1}(u) \cap \mathcal{Q}_i(u)$ for every $u \in V(P)\cap (V_i \setminus \{v\})$. 
\end{observation}
\def\APPENDOBSCLUSTINDUC{
\begin{proof}[Proof of Obs. \ref{obs:cluster-induc}]
We prove it by induction on $i$. For $i=0$, the claim holds vacuously. Assume it holds up to $i-1$, and consider $i$. Fix $v \in V_i$, a path $P \in \mathcal{Q}_i(v)$ and $u \in V(P)\cap V_i$. Let $P'=\Shortcut^{-1}(P)$ and let $(u_j,v)=LastE(P')$. By construction, $P'[\cdot, u_j]\in \mathcal{S}_{i-1}(u_j) \subset \mathcal{Q}_{i-1}(u_j)$. Therefore, by the induction assumption for $i-1$, it holds that $P'[\cdot, u_\ell] \in \mathcal{Q}_{i-1}(u_\ell)$ for every $u_\ell \in V(P'[\cdot, u_j])$. Since $h(P')=h(P)\in Z_i$, by Obs. \ref{obs:cluster-cont}, we also have that $P'[\cdot, u_\ell] \in \mathcal{Q}_{i}(u_\ell)$ for every $u_\ell \in V(P'[\cdot, u_j])$. Since $P=\Shortcut(P')$, there is $u_q \in V(P'[\cdot, u_j])$ such that $P=P'[\cdot, u_q]\circ (u_q,v)$. We conclude that
$P[\cdot, u_\ell]=P'[\cdot, u_\ell] \in \mathcal{Q}_{i-1}(u_\ell) \cup \mathcal{Q}_{i}(u_\ell)$ for every $u_\ell \in V(P'[\cdot, u_q])$ as required. 
\end{proof}
}
\noindent Recall that $i_v$ is the largest index of the paths in $\mathcal{P}^*_i(v)=\{P_1,\ldots, P_\ell\}$ taken into $\mathcal{Q}_i(v)$.
\begin{claim}\label{cl:clust-prob}
Fix a vertex $v \in V_{i-1}$ and let $\mathcal{P}^*_{i-1}(v)=\{P_1,\ldots, P_\ell\}$ where $LastE(P_1)< \ldots < LastE(P_\ell)$. If $|\mathcal{P}^*_{i-1}(v)|\geq c\cdot k f^{1-1/k}n^{1/k}\log n$ for some constant $c>1$, then, w.h.p., $v$ is $i$-clustered and $i_v \leq c\cdot k f^{1-1/k}n^{1/k}\log n$. 
\end{claim}
\def\APPENDCLUSTPROB{
\begin{proof}[Proof of Claim \ref{cl:clust-prob}]
Recall that $h(P) \in Z_{i-1}$ for every $P \in \mathcal{P}^*_{i-1}(v)$ and as the paths in this set are vertex-disjoint (expect for the endpoint $v$), $h(P)\neq h(P')$ for every $P,P' \in \mathcal{P}^*_{i-1}(v)$. Also recall that $Z'_{i-1}(v)=\{h(P)~\mid~ P \in \mathcal{P}^*_{i-1}(v)\}$ and that $v$ is $i$-\emph{clustered} if $|Z'_{i-1}(v) \cap Z_{i}|\geq K_f$. Clearly, $|Z'_{i-1}(v)|=|\mathcal{P}^*_{i-1}(v)|$. Since each $s \in Z_{i-1}$ is sampled into $Z_i$ independently with probability of $p=(f/n)^{1/k}$, in expectation $|Z'_{i-1}(v) \cap Z_{i}|=|Z'_{i-1}(v)|\cdot p\geq c\cdot K_f \cdot \log n$ provided that $|\mathcal{P}^*_{i-1}(v)|\geq c\cdot k f^{1-1/k}n^{1/k}\log n$. 

By a simple application of the Chernoff bound, we get that w.h.p., $|\mathcal{P}^*_{i-1}(v)|\geq c'\cdot k f^{1-1/k}n^{1/k}\log n$ for some constant $c'< c$. Finally for an $i$-clustered vertex $v$, one can apply this argument on the first $c\cdot k f^{1-1/k}n^{1/k}\log n$ paths in $\mathcal{P}^*_{i-1}(v)$ (sorted based on the weight of their last edge). By the above argument, we get that this subset must be hit by at least $k_f$ sampled centers, concluding that w.h.p. $i_v \leq c\cdot k f^{1-1/k}n^{1/k}\log n$.
\end{proof}
}

\begin{lemma}[Size]\label{lem:size}
W.h.p., $|E(H)|=O(k^3 \cdot \log n \cdot f^{1-1/k}\cdot n^{1+1/k}+k^2 f n)$. 
\end{lemma}
\begin{proof}
We focus on phase $i$ and bound the size of $E(H_i)\setminus E(H_{i-1})$. We show that, w.h.p., $|LE_i(v)|=O(k^2 f^{1-1/k}n^{1/k})$ for every $v \in V_{i-1}$. 
Every vertex $v$ becomes unclustered in a unique phase $i$. At this point we add the edges in 
$LE_{i}(v)$ to the spanner $H_i$. By Claim \ref{cl:clust-prob}, the index $i_v$ is at most $O(k f^{1-1/k}n^{1/k}\log n)$. Recall that $LE_{i}(v)$ consists of $v$'s edges to vertices of $V(\Shortcut^{-1}(P_j))$ for every $j \in \{1,\ldots, i_v-1\}$ (see Eq. (\ref{eq:LEiv})). As each path in $P_j$ and $\Shortcut^{-1}(P_j)$ has at most $i$ edges, we have that $|LE_{i}(v)|=O(k^2 f^{1-1/k}n^{1/k}\log n)$. 

In addition, the algorithm adds to $H_i$ the last edges of the path $P$ for each $P \in \mathcal{Q}_i(v)$, for every $v \in V_i$. Note that by Obs. \ref{obs:cluster-induc}, these edges correspond to $E(T_i(s))\setminus E(T_{i-1}(s))$ for every $s \in Z_i$. As $|\mathcal{Q}_i(v)|=K_f$, this adds a total of $O(k f n)$ edges. Overall, we conclude that $|E(H_i)\setminus E(H_{i-1})|=O(k^2 f^{1-1/k}n^{1+1/k}\log n)$. Summing over all $k$ phases provides the final size bound of $O(k^3 f^{1-1/k}n^{1+1/k}\log n)$. 
\end{proof}

We next show that the clustering $\mathcal{C}_i$ obtained at the end of phase $i$ satisfies properties (I-III). The proof of property (IV) is more involved and deferred to Lemma \ref{lem:monotone}.
\begin{lemma}\label{lem:cluster-invariant}
The output clustering $\mathcal{C}_i$ of phase $i$ satisfies properties (I-III), w.h.p.
\end{lemma}
\def\APPENDCLUSTINV{
\begin{proof}[Proof of Lemma \ref{lem:cluster-invariant}]
We prove this by induction on $i$. For $i=0$, properties (I-III) hold vacuously. 
We next prove by induction that $|Z_i|=f^{i/k}\cdot n^{1-i/k}$ in expectation for every $i \in \{0,\ldots, k-1\}$. Since $|Z_0|=V$, the claim holds for $i=0$. Assume it holds up to $i-1$, then in expectation $|Z_i|=p \cdot |Z_{i-1}|=f^{i/k}\cdot n^{1-i/k}$. Note that each vertex $v \in Z_0$ is in $Z_i$ with probability of $p^i=(f/n)^{i/k}$. Using Chernoff bound we have that w.h.p. $|Z_i|=O(f^{i/k}\cdot n^{1-i/k})$. 

We next show that each $T_i(s)$ given by Eq. (\ref{eq:tree}) is a tree of depth at most $i$ for every $s \in Z_i$. Fix a path $P \subset T_i(s)$ such that $P \in \mathcal{Q}_i(v)$ for some $v \in V_i$. It is easy to see $|P|\leq i$. By Obs. \ref{obs:cluster-induc}, we have that for $u \in V(P) \cap V_i$, it holds that $P[\cdot,u] \in \mathcal{Q}_i(u)$. This implies that for each $u \in V(T_i(s))$, there is a unique $s$-$u$ path 
in $T_i(s)$ given by $P[\cdot,u]$. Property (II) follows. Property (III) follows by the fact that the tree paths of vertex $v$ in the clusters of $\mathcal{C}_i$ are given by $\mathcal{Q}_i(v)$, and these paths are, by construction, vertex-disjoint (except for the endpoint $v$). 
\end{proof}
}
\noindent The shortcutting procedure of each path $P \in \mathcal{P}_{i-1}(v)$ immediately implies that:
\begin{observation}\label{obs:shortcut}
$\forall P\in \mathcal{P}_{i-1}(v)$, $W(LastE(\Shortcut(P)))\leq W(e')$ for every $e'=(u,v) \in R_{i-1}$ and $u \in V(P)$.
\end{observation}

\vspace{-12pt}
\subsection{The Stretch Argument} \label{subsec:stretch} \vspace{-5pt}
We next turn to consider the stretch argument, which constitutes the major technical contribution. As we will see the analysis of our FT-variant of the Baswana-Sen algorithm calls for a new graph theoretical characterization. We start by identifying a sufficient condition for $(f,i)$ protection (see Def. \ref{def:protect}). We then show that, w.h.p., this condition holds for every edge in the graph, which will establish that the output subgraph is an $f$-VFT $(2k-1)$-spanner. 

\begin{observation}[Sufficient condition for $(f,i)$ protection]\label{obs:Suff}
Fix an edge $(u,v) \in E$, and let $H \subseteq G$ be a subgraph containing a collection of $(\alpha \cdot f+1)$ paths $\mathcal{P}=\{P_1,\ldots, P_{\alpha \cdot f+1}\}$ between $u$ and $v$ such that:
\begin{itemize}
\item $|P_j|\leq 2i-1$, for every $P_j \in \mathcal{P}$;

\item $(u,v)$ is heavier than every edge in $E(P_j)$, for every $P_j \in \mathcal{P}$;

\item every vertex $x \in V \setminus \{u,v\}$ appears on at most $\alpha$ paths in $\mathcal{P}$;
\end{itemize}
then $(u,v)$ is $(f,i)$ protected in $H$.
\end{observation}
\begin{proof}
Fix a faulty subset $F \subseteq V \setminus \{u,v\}$ where $|F|\leq f$. As each vertex $x \in F$ intersects with at most $\alpha$ paths, we have that $F$ intersects with at most $\alpha \cdot f$ paths. Therefore, there exists a $u$-$v$ path $P_j \in \mathcal{P}$ satisfying that $V(P_j) \cap F=\emptyset$. We have that $\dist_{H\setminus F}(u,v)\leq len(P_{j})\leq (2i-1)W(e)$ as required.
\end{proof}

We next state the main technical stretch argument, which will be shown in a sequence of claims. Most of the analysis is done w.r.t to the set of vertex-disjoint paths $\mathcal{P}_{i-1}(v)$, and only at the end of the arguments we extend them to work w.r.t the shortcut set $\mathcal{P}^*_{i-1}(v)$. 
\begin{mdframed}[hidealllines=true,backgroundcolor=gray!25]
\vspace{-5pt}
\begin{lemma}[Main Stretch Lemma]\label{lem:key}
Fix $i \in \{1,\ldots, k\}$ and let $(u,v) \in R_{i-1}$ such that $u \notin V(\mathcal{P}_{i-1}(v))$. Then, w.h.p., one of the two must hold: either $(u,v)$ is $(f,i)$-protected in $H_i$, or else $(u,v) \in R_i$. 
\end{lemma}
\end{mdframed}

Since $e=(u,v) \in R_{i-1}$, by Def. \ref{eq:Riv} it holds that $u$ and $v$ are $(i-1)$-clustered, and in addition:
\begin{equation}\label{eq:increasing-path-edge}
W(e) > W(e') \mbox{~~for every~~} e' \in \{ E(P) ~\mid~ P \in \mathcal{Q}_{i-1}(u) \cup \mathcal{Q}_{i-1}(v)\}~.
\end{equation}
Let $\ell$ be such that $(u,v)$ is the $\ell^{th}$ edge in $E_{i-1}(v) \subseteq R_{i-1}$, and denote the set $\mathcal{P}_{i-1}(v)$ at the beginning of the $\ell^{th}$ iteration (namely, just before considering the edge $(u,v)$) by $\mathcal{P}_{i-1}(v,u)$. See Fig. \ref{fig:MIS} for an illustration. 
We start with some auxiliary claims to identify cases in which $(u,v)$ is protected. 
\begin{claim}\label{cl:aux1}
If at least $1/4$ of the paths in $\mathcal{Q}_{i-1}(u)$ contains $v$, then $(u,v)$ is $(f,i)$ protected in $H_{i-1}$. 
\end{claim}
\begin{proof}
Let $P_1,\ldots, P_{\ell'}$ be the paths in $\mathcal{Q}_{i-1}(u)$ that contains $v$. Then, the paths $P_1[u,v], \ldots, P_{\ell'}[u,v] \subseteq H_{i-1}$ are vertex-disjoint, and the claim holds as $\ell' \geq K_f/4\geq 3f$ and by Eq. (\ref{eq:increasing-path-edge}).
\end{proof}

\begin{claim}\label{cl:aux2}
If at least half of the paths in $\mathcal{Q}_{i-1}(u)$ \emph{do not intersect} with $\mathcal{P}_{i-1}(v,u)$, then w.h.p., $u \in V(\mathcal{P}_{i-1}(v))$. 
\end{claim}
\begin{proof}
Recall that $\mathcal{S}_{i-1}(u) \subseteq \mathcal{Q}_{i-1}(u)$ is a random sample of $O(\log n)$ paths, each sampled uniformly at random from $\mathcal{Q}_{i-1}(u)$. Since at least half of the paths in $\mathcal{Q}_{i-1}(u)$ do not intersect with $\mathcal{P}_{i-1}(v,u)$ there is a probability of $1/2$ to sample a path that does not intersect $\mathcal{P}_{i-1}(v,u)$. Since the algorithm samples $\Theta(\log n)$ paths independently into  $\mathcal{S}_{i-1}(u)$, w.h.p. at least one of these paths does not intersect $\mathcal{P}_{i-1}(v,u)$. In the latter case, this path is added to $\mathcal{P}_{i-1}(v)$ and $u \in V(\mathcal{P}_{i-1}(v))$. 

\end{proof}
\noindent Since the key stretch lemma considers the case where $u \notin V(\mathcal{P}_{i-1}(v))$, by Claims \ref{cl:aux1} and \ref{cl:aux2}, it remains to consider the remaining case where:
\begin{quote}
\textbf{Case $\star$:} \emph{At least a $3/4$-fraction of the paths in $\mathcal{Q}_{i-1}(u)$ do not contain $v$, and at least half of the paths in $\mathcal{Q}_{i-1}(u)$ intersect $\mathcal{P}_{i-1}(v,u)$. }
\end{quote}

Let $\mathcal{Q}_{i-1}(u,v) \subseteq \mathcal{Q}_{i-1}(u)$ be the subset of paths in $\mathcal{Q}_{i-1}(u)$ that intersect with $\mathcal{P}_{i-1}(v,u)$, and in addition do not contain the vertex $v$. 
Formally, define
$$\mathcal{Q}_{i-1}(u,v)=\{P' \in \mathcal{Q}_{i-1}(u) ~\mid~ v \notin V(P'),  P' \mbox{~intersects~} \mathcal{P}_{i-1}(v,u)\}.$$ 
Hence, by Case $\star$, we have that $|\mathcal{Q}_{i-1}(u,v)|\geq |\mathcal{Q}_{i-1}(u)|/4$. Since $\mathcal{P}_{i-1}(v,u)\subseteq \mathcal{P}_{i-1}(v)$, we also have that the paths of $\mathcal{P}_{i-1}(v,u)$ do not contain $u$.  The next key lemma shows that $\mathcal{P}_{i-1}(v,u)$ and $\mathcal{Q}_{i-1}(u,v)$ contain sufficiently many intersecting pairs of paths. 
This will serve the basis for showing the existence in $H_{i}$ of $3f$ nearly vertex-disjoint paths between $u,v$ of length at most $(2i-1)W((u,v))$, as required.

\begin{mdframed}[hidealllines=true,backgroundcolor=gray!25]
\vspace{-5pt}
\begin{lemma}\label{lem:path-pairs}[Key Structural Lemma]
There exists a collection of $3f$ pairs of paths $(Q'_j,P'_j) \in \mathcal{Q}_{i-1}(u,v) \times \mathcal{P}_{i-1}(v,u)$ such that $V(Q'_j) \cap V(P'_j)\neq \emptyset$, for every $j \in \{1,\ldots, 3f\}$. Moreover, $Q'_{1}, \ldots, Q'_{3f}$ are vertex-disjoint (except for their common endpoint $u$), and similarly, $P'_{1}, \ldots, P'_{3f}$ are vertex-disjoint (except for their common endpoint $v$). 
\end{lemma}
\end{mdframed}
\begin{proof}
\textbf{Defining the path pairs.} We greedily pick the $(Q'_j,P'_j)$ pairs, one by one, in $3f$ iterations as follows. Initially, set $\mathcal{A}_1=\mathcal{Q}_{i-1}(u,v)$ and $\mathcal{B}_1=\mathcal{P}_{i-1}(v,u)$. In every iteration $j \in \{1,\ldots, 3f\}$, the algorithm is given the sets $\mathcal{A}_j, \mathcal{B}_j$ and computes the $j^{th}$ path pair $(Q'_j,P'_j)$. Let $Q'_j$ be some arbitrary path in $\mathcal{A}_j$, and let $P'_j$ be some path in $\mathcal{B}_j$ that intersects $Q'_j$. Then, define the sets 
$\mathcal{A}_{j+1}$ and $\mathcal{B}_{j+1}$, by omitting all paths in $\mathcal{A}_{j}$ that intersect $P'_j$ and omitting $P'_j$ from $\mathcal{B}_{j+1}$. Formally, define 
$$\mathcal{A}_{j+1}=\{Q \in \mathcal{A}_{j} ~\mid~ V(P'_j) \cap V(Q)=\emptyset\} \mbox{~and~} \mathcal{B}_{j+1}=\mathcal{B}_{j} \setminus \{P'_j\}~.$$
\\
\noindent \textbf{Analysis.} We now claim that the algorithm above is valid, we need to show that for every $j \in \{1,\ldots, 3f\}$ it holds that:
\begin{enumerate}
\item the set $\mathcal{A}_j$ is non-empty, and 
\item every path in $\mathcal{A}_j$ intersects $\mathcal{B}_{j}$. 
\end{enumerate}
We start with (1). Since the paths of $\mathcal{Q}_{i-1}(u,v)$ are vertex-disjoint (except for the endpoint $u$), and since we assume (by Case 1) that no path in $\mathcal{P}_{i-1}(v,u)$ contains $u$, we have that each path $P' \in \mathcal{P}_{i-1}(v,u)$ can intersect with \emph{at most} $|P'|\leq k$ many paths in $\mathcal{Q}_{i-1}(u,v)$. Since in each iteration $j$, the algorithm omits $|P'_j|$ paths from $\mathcal{Q}_{i-1}(u,v)$, we have that 
$$|\mathcal{A}_{j+1}|\geq |\mathcal{Q}_{i-1}(u,v)|-k \cdot j~.$$
Therefore, $|\mathcal{A}_{3f}|\geq K_f/4-3kf\geq 2kf \geq 3f$, for $k\geq 2$, as required. 

To see (2), we show by induction on $j\in \{1,\ldots, 3f\}$, that every path in $\mathcal{A}_{j}$ does not intersect with any of the paths in $P'_1,\ldots, P'_{j-1}$. This holds as at the end of each iteration $j$, when adding $P'_j$ we omit from $\mathcal{A}_{j+1}$ all the paths that intersect with $P'_j$. Since each path in $\mathcal{A}_{j} \subseteq \mathcal{Q}_{i-1}(u,v)$ intersects $\mathcal{P}_{i-1}(v,u)$, we deduce that each such path must intersect with the paths of $\mathcal{B}_{j}$. The lemma follows. See Fig. \ref{fig:multi-paths} for an illustration.
\end{proof}

We now need to complete the proof for Case $\star$, and show that either that $(u,v)$ is protected in $H_{i-1}$, or else in $(u,v)\in R_i$. The next lemma is crucial for the stretch argument in the \emph{weighted} setting (i.e., it is not needed for unweighted graphs). 
\begin{claim}\label{cl:aux3}
$(u,v)$ is heavier than any edge in $P'_j \cup Q'_j$ (from Lemma \ref{lem:path-pairs}) for every $j \in \{1,\ldots, 3f\}$. 
\end{claim}
\begin{proof}
Since $(u,v) \in R_{i-1}$ and $Q'_j \in \mathcal{Q}_{i-1}(u)$, we have that $(u,v)$ is heavier than any edge on each of the paths in $\mathcal{Q}_{i-1}(u)$. We next show that $(u,v)$ is also heavier than all edges on $P'_j$. Let $(u_j,v)$ be the last edge of $P'_j$ for every $j \in \{1,\ldots, 3f\}$. Since $(u_j,v) \in R_{i-1}$, and $P'_j[\cdot,u_j]\in \mathcal{S}_{i-1}(u_j)\subset \mathcal{Q}_{i-1}(u_j)$, it holds that $W((u_j,v))\geq W(e')$ for every $e' \in P'_j$. 
Since $P'_j \in \mathcal{P}_{i-1}(v,u)$, by the ordering of the edges in $E_{i-1}(v)$, we have that $W((u,v))>W((u_j,v))$. Therefore, $(u,v)$ is heavier than any edge on $P'_j$.
\end{proof}

\begin{corollary}\label{cor:aux3easy}
If every $P'_j$ (from Lemma \ref{lem:path-pairs}) for $j \in \{1,\ldots, 3f\}$, is taken into $H_i$, then $(u,v)$ is $(f,i)$ protected.
\end{corollary}
\begin{proof}
In this case $P'_j \cup Q'_j \subseteq H_{i-1} \cup H_i$. For every $j \in \{1,\ldots, 3f\}$, let $q_j \in V(Q'_j)\cap V(P'_j)$. Therefore, we have $3f$ walks between $u$-$v$ given by $X_j=Q'_j[u,q_j]\circ P'_j[q_j,v]$. Each walk has at most $|P'_j|+|Q'_j|\leq i+i-1\leq 2i-1$ hops. By combining with Claim \ref{cl:aux3}, we have that $len(X_j)\leq (2i-1)W((u,v))$. 

Since the paths in $\mathcal{Q}_i(u,v)$ are vertex-disjoint (except for the common $u$), each vertex $x \in V \setminus \{u,v\}$ can appear on at most one path $Q'_{j_x} \in \mathcal{Q}_i(u,v)$. In the same manner, since the paths in $\mathcal{P}_i(v,u)$ are vertex-disjoint (except for the common $v$), each vertex $x \in V \setminus \{u,v\}$ can appear on at most one path $P'_{j_x} \in \mathcal{P}_i(v,u)$. Therefore, we conclude that each vertex $x \in V \setminus \{u,v\}$ can appear on at most \emph{two} walks. We have that the $3f$ walks $X_1,\ldots, X_f$ satisfy all properties of Obs. \ref{obs:Suff}, and therefore $(u,v)$ is $(f,i)$ protected.
\end{proof}

\begin{claim}\label{cl:aux4}
If there exists a path $P'_j$ that is not taken into $H_i$ for \emph{some} $j \in \{1,\ldots, 3f\}$ (from Lemma \ref{lem:path-pairs}), then $v \in V_i$ and $(u,v)$ is heavier than all edges of the paths in $\mathcal{Q}_{i}(v)$.
\end{claim}
\begin{proof}
Since by Eq. (\ref{eq:LEiv}), the algorithm adds to $H_i$ all paths in $\mathcal{P}_{i-1}(x)$ for an $i$-unclustered vertex $x$, we can deduce that $v$ is $i$-clustered. 

Since $P'_j \in \mathcal{P}_{i}(v,u)$ (the set $\mathcal{P}_{i}(v)$ just before inspecting $(u,v)$), by the ordering of the edges in $E_{i-1}(v)$, we have that $W((u,v))>W(LastE(P'_j))$. 

Recall that $i_v$ is the largest index of the paths in $\mathcal{P}^*_i(v)=\{P_1,\ldots, P_\ell\}$ taken into $\mathcal{Q}_i(v)$, and fix a path $P_q \in \mathcal{Q}_i(v)$ (where $q \leq i_v$), letting $(u_{q},v)=LastE(P_{q})$.
Since $(u,v)\in R_{i-1}$, it is heavier than all edges on the paths in $\mathcal{Q}_{i-1}(v)$. It therefore remains to consider that case where $P_q \in \mathcal{Q}_i(v) \setminus \mathcal{Q}_{i-1}(v)$, and therefore $LastE(P_q)\in R_{i-1}$.

Since $P'_j$ is not taken into $H_i$, by Eq. (\ref{eq:LEiv}) we have that the path $\Shortcut(P'_j)$ appears \emph{after} the path $P_{i_v}$ in the ordered set $\mathcal{P}^*_{i-1}(v)$. So-far, we have that
$$W((u,v))> W(LastE(P'_j))\geq W(LastE(\Shortcut(P'_j))) > W(LastE(P_q)),$$ 
where the second inequality follows by Obs. \ref{obs:shortcut}.
By Obs. \ref{obs:cluster-induc}, it also holds that $P_q[\cdot , u_q] \in \mathcal{Q}_{i-1}(u_q)$.  Since $LastE(P_q) \in R_{i-1}$, we have that $W(LastE(P_q))$ is the heaviest edge on $P_q$. We conclude that $W(u,v)$ is heavier than any edge on $P_q$ for every $P_q \in \mathcal{Q}_i(v)$. 
\end{proof}

We are now ready to complete the proof of Lemma \ref{lem:key}.
\begin{proof}[Proof of Lemma \ref{lem:key}]
By Claims \ref{cl:aux1}, \ref{cl:aux2}, Cor. \ref{cor:aux3easy} and Claim \ref{cl:aux4}, we have that either $(u,v)$ is $(f,i)$ protected or else, $v \in V_{i}$ and $(u,v)$ is heavier than all edges of the paths in $\mathcal{Q}_{i-1}(v)$. 

By repeating the exact same analysis from the point of view of $u \in V_{i-1}$, we get that w.h.p. either $(u,v)$ is $(f,i)$ protected or else, $u \in V_i$ and $(u,v)$ is heavier than all edges of the paths in $\mathcal{Q}_{i-1}(u)$. Concluding that w.h.p., either $(u,v)$ is $(f,i)$ protected or else, $(u,v) \in R_i$ as required.
\end{proof}

\begin{figure}[h!]
\begin{center}
\includegraphics[scale=0.35]{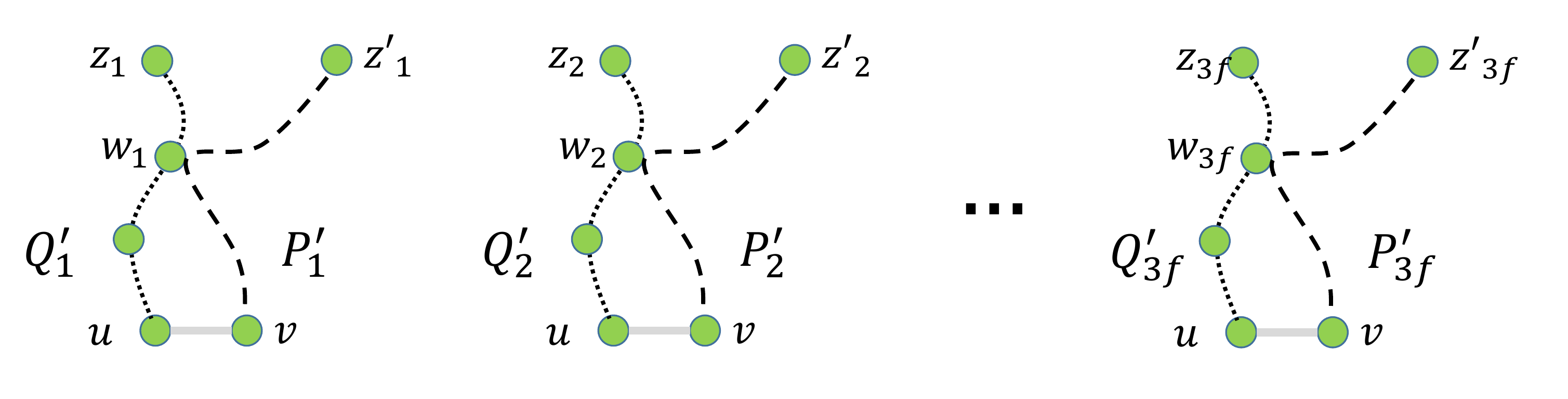}
\includegraphics[scale=0.35]{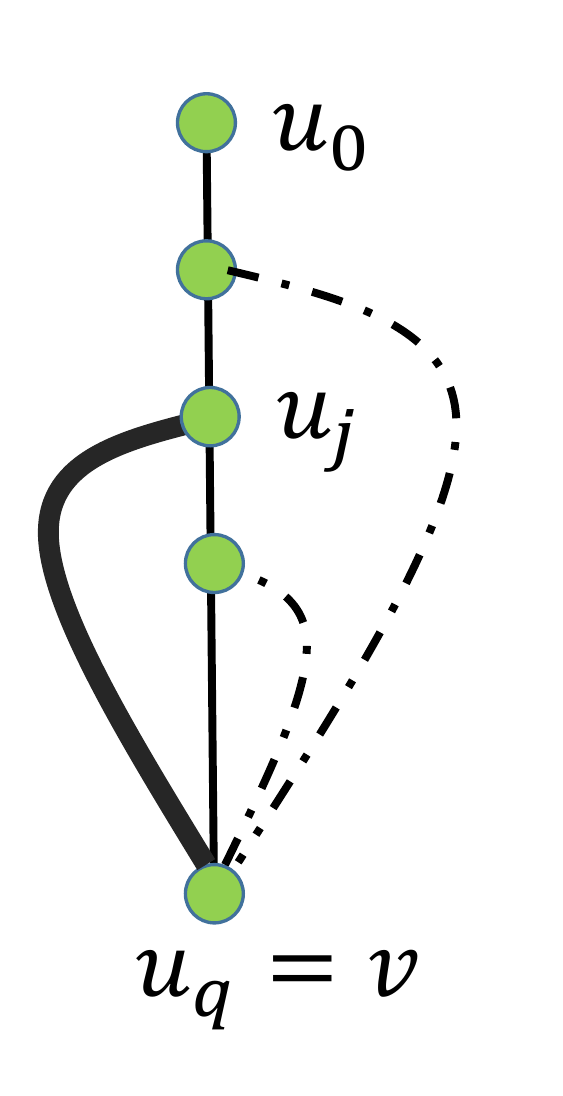}
\caption{\sf Left: An illustration for Lemma \ref{lem:path-pairs}. Right: Illustration of Procedure $\Shortcut$. Shown is a path $P=[u_0,\ldots, u_q=v]$ where $(u_j,v)$ is the \emph{lightest} among all other dashed edges, thus $\Shortcut(P)=[u_0, \ldots, u_j,u_q]$. \label{fig:multi-paths} 
}
\end{center}
\end{figure}

\begin{lemma}\label{lem:monotone}
For every $i \in \{0,\ldots, k-1\}$ and every $v \in V_i$, the edge weights along  each path in $\mathcal{Q}_i(v)$ are monotone increasing (towards $v$). 
\end{lemma}
\def\APPENDMONO{
\begin{proof}[Proof of Lemma \ref{lem:monotone}]
We show the proof by induction on $i \in \{0,\ldots, k-1\}$. For $i=0$, since $\mathcal{Q}_0(v)=\{v\}$ for every $v$, the claim holds vacuously. Assume that for every $j \leq i-1$, the paths of $\mathcal{Q}_j(v)$ are monotone for every $v \in V_j$ and consider the paths of $\mathcal{Q}_i(v)$ for an $i$-clustered vertex $v \in V_i$. 

We first show that the edge weights of any path $P \in \mathcal{P}_{i-1}(v)$ are monotone increasing (i.e., before the shortcutting step). Letting $P=[u_0,\ldots,u_{\ell-1}=u ,u_{\ell}=v]$, then $P[\cdot,u]$ is a path in $\mathcal{S}_{i-1}(u)\subseteq \mathcal{Q}_{i-1}(u)$. Since $(u,v) \in R_{i-1}$, we have that $(u,v)$ is heavier than all other edges on $P[u_0,u]$, therefore $P$ is monotone.

Let $E'=(V(P)\times \{v\}) \cap R_{i-1}$ and denote by $u_j \in V(P)$ as the vertex satisfying that the edge $(u_j,v)$ is the lightest in $E'$. The interesting case is where $u_j \neq u$ as in this case, $\Shortcut(P)=P[u_0,u_j]\circ (u_j,v)$. We show that $\Shortcut(P)$ is monotone as well.
Since $(u_j,v) \in R_{i-1}$, we have that $u_j$ is $(i-1)$-clustered. 
By Obs. \ref{obs:cluster-induc}, we have that $P[\cdot,u_j] \in \mathcal{Q}_{i-1}(u_j)$. Note, however, that might be the case that $P[\cdot,u_j] \notin \mathcal{S}_{i-1}(u_j)$ (i.e., the case where that path $P[\cdot,u_j]$ was not sampled into $\mathcal{S}_{i-1}(u_j)$). Nevertheless, as $P[\cdot,u_j] \in \mathcal{Q}_{i-1}(u_j)$, by the definition of $R_{i-1}$, we have that $W((u_j,v))> W(e')$ for every edge $e' \in P[\cdot,u_j]$. By the induction assumption for $(i-1)$, the weights on $P[\cdot,u_j]$ are monotone increasing (towards $u_j$). The lemma follows.
\end{proof}
}

\begin{lemma}\label{lem:stretch-finallem}
Let $i$ be the largest index in $\{1,\ldots, k\}$ satisfying that $e=(u,v) \in R_{i-1}\setminus R_i$. 
Then, w.h.p., $(u,v)$ is $(f,i)$ protected in $H$.
\end{lemma}
\begin{proof}
Without loss of generality, assume that $v$ stopped being clustered not after $u$. \\
\noindent \textbf{Case 1:} $v \in V_{i-1}\setminus V_i$. 
First assume that $u \notin V(\mathcal{P}_{i-1}(v))$. By Lemma \ref{lem:key}, we have that $(u,v)$ is $(f,i)$ protected in $H$. Now consider the case where $u \in V(\mathcal{P}_{i-1}(v))$. Since $v$ is $i$-unclustered, by Eq. (\ref{eq:LEiv}), the algorithm adds $(u,v)$ to $H_i$ and the claim holds. 
\\ \\
\noindent \textbf{Case 2:} $u,v \in V_i$. Note that by Lemma \ref{lem:key}, the claim follows immediately if either $u \notin   V(\mathcal{P}_{i-1}(v))$ or $v \notin V(\mathcal{P}_{i-1}(u))$. Therefore, we can assume from now on that 
$u \in V(\mathcal{P}_{i-1}(v))$ and $v \in V(\mathcal{P}_{i-1}(u))$. 

Let $\mathcal{P}^*_{i-1}(v)=\{P'_1,\ldots, P'_\ell\}$ and $P'_{i_v}$ be the largest indexed path added to $\mathcal{Q}_{i}(v)$. Let $P_q$ be the (unique) path in $\mathcal{P}_{i-1}(v)$ satisfying that $u \in V(P_q)$. 
Assume first that the path $\Shortcut(P_q)$ appears in the ordered set $\mathcal{P}^*_{i-1}(v)$ not after $P'_{i_v}$. In this case, by Eq. (\ref{eq:LEiv}), the algorithm added $(u,v)$ to $H_i$. Therefore, assume that $\Shortcut(P_q)$ appears in $\mathcal{P}^*_{i-1}(v)$ \emph{after} $P'_{i_v}$. We will show that in such a case the edge $(u,v)$ is heavier than all edges in the paths of $\mathcal{Q}_i(v)$. Fix $P'_j \in \mathcal{Q}_i(v)$. We have that:
$$W((u,v))\geq W(\LastE(\Shortcut(P_q))) > W(\LastE(P'_j))~,$$
where the first inequality follows by Obs. \ref{obs:shortcut}, and the last inequality following by the fact that $P'_j$ appears in $\mathcal{P}^*_{i-1}(v)$ strictly before $\Shortcut(P_q)$. By Lemma \ref{lem:monotone}, $W(\LastE(P'_j))$ is the heaviest edge on $P'_j$. Altogether, $(u,v)$ is heavier than any edge on every $P'_j \in \mathcal{Q}_i(v)$. In a symmetric manner, one can show that $(u,v)$ is heavier than any edge on every $P'_j \in \mathcal{Q}_i(u)$. 
Therefore, we get that $(u,v) \in R_i$, leading to a contradiction. 
\end{proof}

\begin{corollary}\label{cor:final-spanner}
Every edge $(u,v) \in E(G)$ is $(f,k)$ protected in $H$, thus $H$ is an $f$-FT $(2k-1)$ spanner.
\end{corollary}
\begin{proof}
Since $R_0=E$ and $R_k=\emptyset$ (as $V_{k}=\emptyset$), for every edge $e \in E(G)$ there exists $i \in \{1,\ldots, k\}$ satisfying that $e \in R_{i-1}\setminus R_i$. The claim follows by Lemma \ref{lem:stretch-finallem}.
\end{proof}

\subsection{Implementation Details and Running Time Analysis}\label{sec:spanner-imp}
We next turn to analyze the running time of Alg. $\VertexFTSpanner$ in the sequential, parallel and distributed settings. For each of these settings we show that the algorithm can be implemented in nearly optimal time. 

\paragraph{The Sequential Setting.} We focus on phase $i$ and show that it can be implemented in $\widetilde{O}(m)$ time, w.h.p. We start with Step (1). The computation of a sampled set $\mathcal{S}_{i-1}(v)$ can be implemented in $\widetilde{O}(|\mathcal{Q}_{i-1}(v)|)=\widetilde{O}(f)$. Therefore, Step (1.1) is implemented in $\widetilde{O}(f\cdot n)$ time. Next, consider the computation of the set $\mathcal{P}_{i-1}(v)$. For each neighbor $u \in N(v)\cap V_{i-1}$, the algorithm iterates over the sampled paths $\mathcal{S}_{i-1}(u)$ in an attempt to find a path that is vertex-disjoint from all the current paths in $\mathcal{P}_{i-1}(v)$. By storing the vertices of 
$\mathcal{P}_{i-1}(v)$ in a (dynamic) hash table, this can be done in $\widetilde{O}(|P|)$ time for each path $P \in \mathcal{S}_{i-1}(u)$. Since $|P|\leq i$ and $|\mathcal{S}_{i-1}(u)|=O(\log n)$, the $j^{th}$ iteration takes $\widetilde{O}(1)$ time. As we have at most $\widetilde{O}(deg(v))$ iterations and summing over all vertices, this step is implemented in $\widetilde{O}(m)$ time, w.h.p. 

For a path $P \in \mathcal{P}_{i-1}(v)$, applying procedure $\Shortcut(P)$ takes $O(|P|\log n)$ time, therefore the computation of the set $\mathcal{P}^*_{i-1}(v)$ of Step (1.2.d) can be done in $O(|V(\mathcal{P}_{i-1}(v))|\log n)$. Clearly, $|V(\mathcal{P}^*_{i-1}(v))|\leq k \cdot \deg(v,G)$ (as the algorithm adds to $\mathcal{P}_{i-1}(v)$ at most one path in $\mathcal{S}_{i-1}(u)$ for each $u \in N(v)$). Therefore, summing over all the vertices, the computation of the sets $\bigcup_{v \in V_{i-1}}\mathcal{P}^*_{i-1}(v)$ takes $\widetilde{O}(m)$ time, as well. 

We next turn to Step (2) where the clusters $\mathcal{C}_i$ are computed. Computing the center set $Z_i$ is done in $\widetilde{O}(|Z_{i-1}|)$ time. For each $v \in V_{i-1}$, the computation of the sets $\mathcal{Q}_{i}(v)$ and $LE_i(v)$ can be computed in time $\widetilde{O}(|V(\mathcal{P}_{i-1}(v))|)$. Therefore, taking $\widetilde{O}(m)$ time, for all the $(i-1)$-clustered vertices. 

Finally, we consider Step (3) and show that the edge set $R_i$ to be computed in $\widetilde{O}(m)$ time, as follows. For every vertex $v$, we assume that the algorithm stores its incident edges $E_{i-1}(v) \subseteq R_{i-1}$ in increasing edge weights. We also store explicitly its paths $\mathcal{Q}_i(v)$. Let $e^*$ be the heaviest edge in $LastE(P'_j)$ for $P'_j \in \mathcal{Q}_i(v)$. Then, the set $E_{i}(v)=\{ (u,v) \in E_{i-1}(v) ~\mid~ W((u,v))> W(e^*), u \in V_i\}$ can be defined in $\widetilde{O}(\deg(v,G))$. The final set $R_i$ is given by $R_i=\bigcup_{v \in V_i}E_i(v) \setminus H_i$, leading to a total running time of $\widetilde{O}(m)$ as desired. Theorem \ref{thm:linear-time-opt-randomized} follows by combining with Lemma \ref{lem:size} and Cor. \ref{cor:final-spanner}. 

\paragraph{Distributed Implementation.} We consider the standard congest model of distributed computing \cite{peleg2000distributed}. In this model, the algorithm works in synchronous rounds, and in each round, every neighboring pair can exchange $O(\log n)$ bits of information. We show that each phase $i$ can be implemented in $\widetilde{O}(1)$ congest rounds, which together with Lemma \ref{lem:size} and Cor. \ref{cor:final-spanner} establishes Theorem \ref{thm:linear-time-opt-randomized-distributed}. 

Start with Step (1). We assume that at the beginning of the phase $i$, each vertex $v$ knows if it is in $V_{i-1}$. In the latter case, it also knows its path collection $\mathcal{Q}_{i-1}(v)$, and its incident edges in $R_{i-1}$. 
Each $v \in V_{i-1}$ locally computes its sampled set $\mathcal{S}_{i-1}(v) \subseteq \mathcal{Q}_{i-1}(v)$, and sends the path information  $\mathcal{S}_{i-1}(v)$ to its neighbors. Since $|V(\mathcal{S}_{i-1}(u))|=O(i \log n)$, this can be done in $O(i\log n)$ rounds. 

From this point on, the vertex $v$ locally implements the MIS computation of Step (1.2) as it obtains all the necessary information, i.e., the path collection $\bigcup_{u \in N(v)\cap V_{i-1}}\mathcal{S}_{i-1}(u)$. As $v$ knows the weight of its incident edges, it can also apply the shortcut procedure and obtain the final set $\mathcal{P}^*_{i-1}(v)$. Altogether, Step (1) is implemented in $\widetilde{O}(1)$ rounds. To implement Step (2), each center $s \in Z_{i-1}$ locally samples itself into $Z_i$ with probability of $p$. Every sampled center $s$ notifies its cluster members. Recall that the collection of trees $\{T_{i-1}(s), s \in S\}$ are of depth at most $i$. In addition, these tree are vertex independent (and therefore also edge-disjoint). 
\begin{observation}\label{obs:efficient-dist-trees}
Each vertex $s$ can send an $O(\log n)$-bit messages to its cluster $C_{i-1}(s)$ within $O(i)$ rounds, in parallel for every $s \in Z_{i-1}$.
\end{observation}
\begin{proof}
We show the each edge $e=(u,v)$ can appear on at most two trees in $\{T_{i-1}(s), s \in S\}$. To see this, observe the by the vertex-independence property, we have that $u$ can be the parent of $v$ in at most one tree. In the same manner, $v$ can be the parent of $u$ in at most one tree. Altogether, $(u,v)$ appears on at most two trees. This allows the sources of $Z_{i-1}$ to broadcast a message on all the trees $\{T_{i-1}(s), s \in S\}$ in parallel within $O(i)$ rounds (each edge needs to send at most two messages, which can be simply done by simulating each broadcast round using two rounds).
\end{proof}
This allows all $(i-1)$-clustered vertices to learn which of their centers are sampled into $Z_i$. Every vertex $u \in V_{i-1}$ can then send to each neighbor $v \in N(v)$ the list of sampled centers $\{h(P) ~\mid~ P \in \mathcal{S}_{i-1}(u)\}$. As $|\mathcal{S}_{i-1}(u)|=O(\log n)$, this can be done within a single communication round. Each vertex $v$ can then locally computes its nearest $K_f$ sampled paths $\mathcal{Q}_{i}(v) \subseteq \mathcal{P}^*_{i-1}(v)$. In the case where $\mathcal{P}^*_{i-1}(v)$ contains less than $K_f$ sampled paths, $v$ is declared as $i$-unclustered. Overall, in the output format of the clustering, each $i$-clustered vertex $v$ knows $\mathcal{Q}_{i}(v)$ as desired. Finally, the edge sets $LE_i(v)$ and $\{(u,v) \in R_i, u \in N(v)\}$ can be locally defined by $v$ with no farther communication. This completes the implementation details of phase $i$, which can indeed be implemented in $\widetilde{O}(1)$ rounds as desired.

\vspace{-8pt}\section{Parallel Implementations (Proof of Theorem \ref{thm:linear-time-opt-randomized-PRAM})}\label{sec:parallel-imp}
The PRAM algorithm consists of a sequence of rounds, where each round consists of a number of computations (e.g., memory access or RAM operations) that are independent of each other, and can be performed in parallel. The total number of rounds is denoted as the \emph{depth} of the computation, and the total number of computations, over all rounds, is denoted by the \emph{work} of the computation. 

The model of parallel computation considered in this section is the CRCW PRAM. This model supports concurrent read and write operations. If multiple processors write to the same entry, an arbitrary one takes effect. 
Throughout, we consider unweighted $n$-graph with $m$ edges, and the algorithms will be using $\widetilde{O}(m)$ space and processors. It might be instructive to allocate a processor $p_{v,u}$ for each edge $(v,u)$ (viewed as a directed edge) that is responsible for the $(v,u)$-``part" in the computation needed for vertex $v$. Throughout, we assume that the graph is unweighted (this simplifies the construction and analysis of Sec. \ref{sec:spanners-general-meta}). 
Throughout, we consider unweighted $n$-graph with $m$ edges, and the algorithms will be using $\widetilde{O}(m)$ space and processors. It might be instructive to allocate a processor $p_{v,u}$ for each edge $(v,u)$ (viewed as a directed edge) that is responsible for the $(v,u)$-``part" in the computation needed for vertex $v$. Throughout, we assume that the graph is unweighted (this simplifies the construction and analysis of Sec. \ref{sec:spanners-general-meta}). 

%
%
%
The key challenge is in implementing Step (1) of Alg. $\VertexFTSpanner$, i.e., the computation of a maximal independent set of vertex-disjoint paths, connecting $v$ to its adjacent clusters in $\mathcal{C}_{i-1}$. Note that this computation is part of the \emph{local} computation in the distributed implementation, therefore this challenge arises only in the parallel setting. 

Letting $N'(v)=\{u_1,\ldots, u_\ell\}$ be the $(i-1)$-clustered neighbors of $v$, then the sequential implementation iterates over the paths of $\mathcal{P}=\bigcup_{j=1}^{\ell}\mathcal{S}_{i-1}(u_j)$ and greedily adds \emph{independent} paths to the current set $\mathcal{P}_{i-1}(v)$. 
Computing the set $\mathcal{P}_{i-1}(v)$ boils down into a lexicographic-first MIS computation according to some ordering of the paths.  In this view, two paths $P$ and $P'$ are \emph{independent} if $V(P)\cap V(P')=\emptyset$, otherwise, $P,P'$ are considered to be \emph{neighbors}. We first provide a modification for the meta-algorithm that will be a more convenient starting point for the parallel implementation.

\paragraph{Modified Step 1.} Focusing on $v \in V_{i-1}$, initially set $\mathcal{P}_{i-1}(v)=\mathcal{Q}_{i-1}(v)$. 
Let $\mathcal{P}=\bigcup_{j=1}^{\ell}\mathcal{S}_{i-1}(u_j)$ and consider the paths in $\mathcal{P}$ in some arbitrary ordering. Iterate over the paths $P$ of $\mathcal{P}$ in a fixed arbitrary ordering, and add a path $P \circ (t(P),v)$ to output set of independent paths $\mathcal{P}_{i-1}(v)$, only if $P$ does not intersect with the current set of paths already taken into $\mathcal{P}_{i-1}(v)$.  That is, whereas in the description of Alg. $\VertexFTSpanner$, the algorithm iterates over $u_j \in N'(v)$, and considers the paths in $\mathcal{S}_{i-1}(u_j)$ one after the other, here the paths $\mathcal{P}$ are considered in any arbitrary ordering. 

The size analysis is unaffected by this modified step, and we only need to reprove Lemma \ref{lem:key}.
Fix an edge $(u,v)$ where $u=u_j$ is the $j^{th}$ neighbor of $v$ in $V_{i-1}$. 
For every $\ell \in \{1,\ldots, |\mathcal{S}_{i-1}(u)|\}$, let $\mathcal{P}^\ell_{i-1}(v,u)$ be the current set of independent paths just before inspecting the $\ell^{th}$ path in $\mathcal{S}_{i-1}(v)$ denoted by $P^{\ell}(u)$. We consider a weaker version of Claim \ref{cl:aux2} by showing:
\begin{claim}\label{cl:aux2-parallel}
If at least half of the paths in $\mathcal{Q}_{i-1}(u)$ \emph{do not intersect} with $\mathcal{P}^\ell_{i-1}(v,u)$, then with \emph{constant probability} $P^{\ell}(u)$ is added to $\mathcal{P}_{i-1}(v))$. 
\end{claim}
\begin{proof}
Since the $\ell$'s path is sampled uniformly from $\mathcal{Q}_{i-1}(u)$ into $\mathcal{S}_{i-1}(u)$, there is a probability of $1/2$ that the sampled path does not intersect $\mathcal{P}^\ell_{i-1}(v,u)$. In the latter case, it is indeed added to the set $\mathcal{P}_{i-1}(v)$.
\end{proof}
We then define Case $(\star,\ell)$ in a similar manner as in the meta-algorithm but restricted to the set $\mathcal{P}^\ell_{i-1}(v,u)$. 
\begin{quote}

\textbf{Case $(\star,\ell)$:} \emph{At least a $3/4$-fraction of the paths in $\mathcal{Q}_{i-1}(u)$ do not contain $v$, and at least half of the paths in $\mathcal{Q}_{i-1}(u)$ intersect $\mathcal{P}^\ell_{i-1}(v,u)$. }
\end{quote}
The proof of Lemma \ref{lem:path-pairs} is exactly the same, concluding that:
\begin{lemma}\label{lem:path-pairs-parallel}
If Case $(\star,\ell)$ holds then there exists a collection of $3f$ pairs of paths $(Q'_a,P'_a) \in \mathcal{Q}_{i-1}(u,v) \times \mathcal{P}^\ell_{i-1}(v,u)$ such that $V(Q'_a) \cap V(P'_a)\neq \emptyset$, for every $a \in \{1,\ldots, 3f\}$.  Therefore, $(u,v)$ is $(f,i)$ protected. 
\end{lemma}
We are now ready to prove Lemma \ref{lem:key} for this modified step. 

\begin{proof}[Proof of Lemma \ref{lem:key} for the Modified Step 1.]
By Lemma \ref{lem:path-pairs-parallel} and Claim \ref{cl:aux1}, it is sufficient to restrict attention to the case where Case $(\star,\ell)$ does not hold for \emph{any} value of $\ell \in \{1,\ldots, |\mathcal{S}_{i-1}(v)|\}$, and in addition, $v$ does not appear on at least $1/4$ of the paths in $\mathcal{Q}_{i-1}(u)$. 

Therefore, we conclude that at least half of the paths in $\mathcal{Q}_{i-1}(u)$ \emph{do not intersect} with $\mathcal{P}^\ell_{i-1}(v,u)$ for every $\ell \in \{1,\ldots, |\mathcal{S}_{i-1}(u)|\}$. By Claim \ref{cl:aux2-parallel}, we then have that each $P^{\ell}(u) \circ (u,v)$ is added to $\mathcal{P}_{i-1}(v)$ with constant probability. Since each path $P^{\ell}(u)$ corresponds to an independent uniform sample in $\mathcal{Q}_{i-1}(u)$, w.h.p., one of these $P^{\ell}(u)$ is added to $\mathcal{P}_{i-1}(v)$, and therefore $u \in V(\mathcal{P}_{i-1}(v))$. 
\end{proof}

\paragraph{Parallel lexicographic-first MIS.}   Our parallel implementation of the modified Step 1 is based upon the influential work, Blelloch, Fineman and Shun \cite{BlellochFS12}. They showed that one can implement the lexicographic-first MIS w.r.t.\ a random ordering $\pi$ in $O(\log^2 n)$ depth and linear work. Shun et al. \cite{ShunGBFG15} also showed that one can compute a random permutation in linear-work and logarithmic depth. The MIS algorithm based random ordering of \cite{BlellochFS12} implements $O(\log n)$ parallel computations of independent sets w.r.t $\pi$ as follows:
\vspace{-0.1in}

\begin{algorithm}
\caption{Parallel greedy algorithm for maximal independent set (taken from \cite{BlellochFS12})} \label{alg:fully-parallel}
\begin{algorithmic}[1]
\Procedure{Parallel-greedy-MIS}{$G = (V,E)$, $\pi$ }
\If{ $|V| = 0$ } \Return $\O$
\Else
\State let $B$ be the set of vertices in $V$ with no earlier neighbors (based on $\pi$)
\State $V' = V \setminus (B \cup  N(B))$
\State \Return $B$ $\cup$ \Call{Parallel-greedy-MIS}{ $G[V']$, $\pi$ }
\EndIf
\EndProcedure
\end{algorithmic}
\end{algorithm}

\subsection{Parallel Implementation of the Modified Step 1}
The  sets $\mathcal{S}_{i-1}(v)$ for $v \in V_{i-1}$ are obtained by taking $O(\log n)$ independent samples uniformly at random from $\mathcal{Q}_{i-1}(v)$. This can be done in $\widetilde{O}(m)$ work and $\widetilde{O}(1)$ depth.

We focus on a vertex $v \in V_{i-1}$ and explain how to compute a set of independent $v$-paths $\mathcal{P}_{i-1}(v)$ using total work $\widetilde{O}(deg(v))$ and depth $\widetilde{O}(1)$. 
This is done by implementing the greedy MIS algorithm using a random ordering $\pi$ on the paths in $\mathcal{P}=\bigcup_{u \in N'(v)}\mathcal{S}_{i-1}(u)$. Recall that initially, we set $\mathcal{P}_{i-1}(v)=\mathcal{Q}_{i-1}(v)$. We first employ a cleanup step that filters out the paths in $\bigcup_{u \in N(v)\cap V_{i-1}}\mathcal{S}_{i-1}(u)$ that intersect with $\mathcal{Q}_{i-1}(v)$. Using hash table for the vertices in $V(\mathcal{Q}_{i-1}(v))$, each processor $p_{v,u_j}$, can filter out the paths in $\mathcal{S}(u_j)$ that intersect with $V(\mathcal{Q}_{i-1}(v))$. This can be done in $\widetilde{O}(1)$ depth in $\widetilde{O}(m)$ works as $|V(\mathcal{S}(u_j))|=\widetilde{O}(1)$.

At this point, we assume that $\mathcal{P}$ is the remaining set (that does not intersect $\mathcal{Q}_{i-1}(v)$), and now we get a clean MIS instance to solve over these paths. Our goal is to implement Alg. \ref{alg:fully-parallel} on a conflict graph in which each path $P \in \mathcal{P}$ is a node and an edge $(P,P')$ exists if $V(P)\cap V(P')\neq \emptyset$. Since $|\mathcal{P}|$ might be as large as $deg(v)$, and since each path in $\mathcal{P}$ might intersect with many other paths, the conflict graph might have $O(deg^2(v))$ edges, leading to a total work of $O(deg^2(v)\cdot n)$. We show that this algorithm can be in fact be implemented in $O(deg(v))$ work even if the max degree of the conflict graph induced by $\mathcal{P}$ is high. 
Implementing Alg. \ref{alg:fully-parallel} of \cite{BlellochFS12} boils down into efficiently computing:
\begin{enumerate}
\item the subset of paths $\mathcal{IS} \subseteq \mathcal{P}$ with no earlier intersecting paths (based on $\pi$), 
\item the remaining set of paths $\mathcal{U}=\{ P \in \mathcal{P} ~\mid~ V(P) \cap V(\mathcal{IS})=\emptyset\}$. 
\end{enumerate}

\paragraph{Computing the independent paths $\mathcal{IS}$.} For each path $P \in \mathcal{P}$, let $\pi(P)$ be the index of $P$ in the ordering $\pi$. We keep a hash table of tuples $(w, \pi(P))$ for each $P \in \mathcal{P}$ and $w \in V(P)$, and use the semi-sorting algorithm on this hash (see \cite{BastH91}, Lemma 5.4 of \cite{BaswanaS:07}). This allows us to compute for each $w \in V(\mathcal{P})$, the index $i_w$ that corresponds to the earliest path in $\mathcal{P}$ that contains $w$, i.e., $i_w=\min\{ \pi(P) ~\mid~ w \in V(P)\}$. We assign a processor $p_j$ for the $j^{th}$ path $P$ such that $\pi(P)=j$. This processor determines if $P \in \mathcal{IS}$, by iterating over all vertices $w \in V(P)$, and adding $P$ to the list only if $j=i_w$ for every $w \in V(P)$. 
This can be done in $\widetilde{O}(deg(v))$ work and poly-logarithmic depth. 

\paragraph{Computing the remaining paths $\mathcal{U}$.} Store a hash table of $V(\mathcal{IS})$. Each processor $p_j$ adds $P_j$ to $\mathcal{U}$ only if none of the vertices in $P_j$ appears in this hash table. 
As this takes $\widetilde{O}(|P_j|)=\widetilde{O}(1)$ time per processor, and since we need $\widetilde{O}(deg(v))$ processors, overall this computation takes $\widetilde{O}(deg(v))$ work and $\widetilde{O}(1)$ depth. This completes the description of computing the independent set of paths $\mathcal{P}_{i-1}(v)$. Since we restrict attention to the unweighted case, there is now no need for the applying the shortcut procedure, and we can safely assign $\mathcal{P}^*_{i-1}(v)=\mathcal{P}_{i-1}(v)$.

Since Alg. \ref{alg:fully-parallel} has $\widetilde{O}(1)$ depth and each recursive call is implemented in $\widetilde{O}(m)$ work, this completes the desired complexity bounds for Step 1. 

\subsection{Parallel Implementation of Step 2}
The sampling of the centers $Z_i \subseteq Z_{i-1}$ can be done in linear time and constant depth. 
Again we focus on a single vertex $v$ and explain how to compute its cluster-path set $\mathcal{Q}_{i}(v) \subseteq \mathcal{P}_{i-1}(v)$ using $\widetilde{O}(deg(v))$ processors and depth $\widetilde{O}(1)$. 
In the case where $v$ is $(i-1)$-unclustered this set is empty, and otherwise it consists of $K_f$ vertex disjoint paths (except for the endpoint $v$). The input is given by $\mathcal{P}^*_{i-1}(v)$ where each $O(\log n)$-length path $P \in \mathcal{P}^*_{i-1}(v)$ is indexed by the neighbor $u_j \in N(v)$ such that $P \in \mathcal{S}_{i-1}(u_j)$.  We assign a value $val(P)=1$ if $h(P)$ is sampled into $Z_i$, and $val(P)=0$ otherwise. We apply a standard sorting on this path list based on $val(P)$, and count the number of paths $\mathcal{P}^*_{i-1}(v)$ with $val(P)=1$. If there are less than $K_f$ such paths, then $v$ is marked as $i$-unclustered. Otherwise, the first $K_f$ paths with $val(P)=1$ are added to $\mathcal{Q}_i(v)$. This completes the description of the second step.  

\paragraph{Parallel Implementation of Step 3.} Since $G$ is unweighted, $LE_i(v)$ is empty for $i$-clustered vertices. 
For an $i$-unclustered vertex $v$, the edge set $LE_i(v)$ consists of all $(\{v\} \times V')\cap E$ edges where $V'=\bigcup_{P \in \mathcal{P}_{i-1}(v)}V(P)$. Since the edges of $G$ are stored by hash table and as $|V'|=\widetilde{O}(deg(v))$, the edge set $LE_i$ can be defined using $|\mathcal{P}_{i-1}(v)|\leq deg(v)$ processors. Each processor would be responsible for one path $P \in \mathcal{P}_{i-1}(v)$, and will mark the edges in $\{v\}\times V(P)$ using $O(|P|)$ hash accesses. The edges of $R_i$ are simply edges in $(V_i \times V_i) \cap E$ that not in $H_i$. Thus the computation of this set can be done in $O(m)$ work and constant depth. 
Since the implementation of the parallel MIS algorithm with respect to the random ordering $\pi$ is \emph{equivalent} to computing the lexicographic-first MIS w.r.t $\pi$, the correctness follows immediately. This completes the proof of Theorem \ref{thm:linear-time-opt-randomized-PRAM}(1).

\section{Efficient Constructions of Vertex Connectivity Certificates}\label{sec:certificate}

\begin{definition}[Connectivity Certificates]\label{def:conn}
Given a graph $G=(V,E)$, a certificate of $\lambda$-vertex connectivity is a spanning subgraph $H \subseteq G$ such that $H$ is $\lambda$-vertex connected if and only if $G$ is $\lambda$-vertex connected. A certificate is 
\emph{sparse} if it contains $O(\lambda)$ edges. 
\end{definition}
As observed in \cite{Par19}, any $f$-EFT $(2k-1)$ spanner $H$ is also a $(f+1)$-edge connectivity certificate. The same proof holds for vertex faults, for completeness we show:
\begin{observation}
Any vertex $f$-FT $(2k-1)$ spanner $H \subseteq G$ is also an $(f+1)$-vertex connectivity certificate.
\end{observation}
\begin{proof}
We show that for every $F \subseteq f$, $u$-$v$ are connected in $H \setminus F$ iff they are connected in $G \setminus F$. This immediately follows by the definition of $f$-FT $(2k-1)$ spanners. By Menger's theorem, a pair of vertices $u,v$ are $\lambda$-vertex connected in $G$ iff $u$ and $v$ are connected in $G \setminus F$ for every subset $F \subseteq V \{u,v\}$, $|F|\leq \lambda-1$. Therefore, we get that $u$-$v$ are $(f+1)$-connected in $H$ iff they are $(f+1)$-connected connected in $G$. 
\end{proof}

\begin{corollary}
An $f$-VFT $O(\log n)$-spanner with nearly optimal sparsity provides an $(f+1)$-vertex connectivity certificates with nearly optimal sparsity. 
\end{corollary}
\noindent This in particular implies that computing nearly sparse $\lambda$-vertex connectivity certificates is a \emph{local} (rather than a global) task. Setting $k=O(\log n)$ and $\lambda=f$ in Theorem \ref{thm:linear-time-opt-randomized}, \ref{thm:linear-time-opt-randomized-distributed} and \ref{thm:linear-time-opt-randomized-PRAM} provides an $f$-vertex connectivity certificates with nearly optimal sparsity in nearly optimal time, in all three settings of computations: sequential, parallel and distributed. We have: 
\begin{corollary}\label{cor:cert}
There is a randomized algorithm that for any $n$-vertex graph $G=(V,E)$ and integer parameter $\lambda$ computes a $\lambda$-vertex connectivity certificates with $\widetilde{O}(\lambda n)$ edges in nearly optimal time complexity in the sequential, distributed and parallel settings. 
\end{corollary}

\section{Deterministic Constructions}\label{sec:derand}
In this section, we turn to consider the deterministic constructions of FT-spanners by means of derandomizing Alg. 
$\VertexFTSpanner$. Technic-wise, this is quite similar to the standard derandomization of the Baswana-Sen algorithm (both in the sequential and the distributed settings). This derandomization increases the sequential and the distributed running times by factor of $O(f)$. Providing deterministic constructions for FT-spanners that nearly match their randomized counterparts is an interesting open problem. 

\subsection{The Sequential Setting}
There are only two parts of Alg. $\VertexFTSpanner$ that relay on randomness. The first is in the definition of the sampled set $\mathcal{S}_{i-1}(v) \subseteq \mathcal{Q}_{i-1}(v)$ (Step (1.2)). The second is in the selection of the centers $Z_i \subset Z_{i-1}$ (Step (2.1). To eliminate the randomness of the first part, namely, Step (1.2), we simply omit it and set $\mathcal{S}_{i-1}(v)=\mathcal{Q}_{i-1}(v)$. The only purpose of computing the sampled set $\mathcal{S}_{i-1}(v)$ was for the sake of obtaining a nearly linear running time. 
Using the sets $\mathcal{Q}_{i-1}(u)$ instead, increases the running time to $\widetilde{O}(f m)$ as each set $\mathcal{Q}_{i-1}(u)$ consists of $K_f=20k f$ many paths, and it is processed by each neighbor $v$ of $u$. We next explain how to derandomize the second part of selecting the cluster centers $Z_i$. 

\paragraph{Computing Cluster Centers Deterministically.}
Recall that $\mathcal{P}^*_{i-1}(v)$ is a collection of vertex-disjoint paths from $Z_{i-1}$ to $v$, and that $Z'_{i-1}(v)=\{h(P)~\mid~ P \in \mathcal{P}^*_{i-1}(v)\}$ where $Z'_{i-1}(v)\subseteq Z_{i-1}$. 
In the randomized algorithm, the center set $Z_i \subset Z_{i-1}$ is computed by sampling each vertex in $Z'_{i-1}$ independently with probability of $p=(f/n)^{1/k}$. As can be seen in the analysis, it suffices for this (random) selection to satisfy two properties for every $i \in \{1,\ldots, k-1\}$:
\begin{enumerate}
\item $|Z_i|=|Z_{i-1}|\cdot p$, and 
\item for each vertex $v$ with $|Z'_{i-1}(v)|\geq c\cdot K_f \cdot \log n/p$, it holds that $|Z'_{i-1}(v) \cap Z_i|\geq K_f$. 
\end{enumerate}
Similarly to the standard Baswana-Sen, as observed in \cite{RodittyTZ05,Censor-HillelPS17,GhaffariK18}, the above two properties can be formulated as an instance to the hitting set problem, for which efficient deterministic solutions exist, such as the following:
\begin{theorem}[Deterministic Hitting Sets, Lemma 6 of \cite{AlonCC19}]\label{thm:det-hitting}
Let $1 \leq \Delta \leq n$ and $1 \leq \ell < poly(n)$ be two integers. Let $R_1,\ldots, R_\ell \subseteq R$ be subsets of vertices satisfying $|R_i|=\Theta(\Delta\cdot \log \ell)$ for every $i \in \{1, \ldots, \ell\}$. Then there is a deterministic algorithms that in $\widetilde{O}(\ell \Delta)$ time computes a set $R^* \subseteq R$ such that for every $i \in \{1,\ldots, \ell\}$ it holds that $R^* \cup R_i \neq \emptyset$ and $|R^*|\leq |R|/\Delta$. 
\end{theorem}
In our setting, we require each the intersection of each set $R_j$ with the hitting set $R^*$ to be sufficiently large (rather than non-empty). However, it is easy to formulate this requirement by slightly modifying the hitting set instance. Specifically, as an immediate corollary of Theorem \ref{thm:det-hitting}, we have:
\begin{corollary}\label{cor:fhittingsets}
Let $1 \leq \Delta,\beta \leq n$ and $1 \leq \ell < poly(n)$ be two integers. Let $R_1,\ldots, R_\ell \subseteq R$ be subsets of vertices satisfying $|R_i|=\Theta(\beta \cdot \Delta\cdot \log \ell)$ for every $i \in \{1, \ldots, \ell\}$. Then there is a deterministic algorithms that in $\widetilde{O}(\beta \cdot \ell \cdot \Delta)$ time computes a set $R^* \subseteq R$ such that for every $i \in \{1,\ldots, \ell\}$ it holds that $|R^* \cup R_i| \geq \beta$ and $|R^*|\leq |R|/\Delta$.
\end{corollary}
\begin{proof}
Partition each $R_i$ set into $\beta$ disjoint sets $R_{i,1},\ldots, R_{i,\beta}$ of roughly equal-size of $\Theta(\Delta \log \ell)$. Apply the algorithm of Theorem \ref{thm:det-hitting} on the collection of $\ell\cdot \beta$ sets $\mathcal{R}=\{R_{i,j} ~\mid~ i \in \{1,\ldots,\ell\}, j \in \{1,\ldots, \beta\}\}$. This results in a set $R^*$ of cardinality $|R|/\Delta$ such that $R^* \cap R_{i,j}\neq \emptyset$ for every $R_{i,j} \in \mathcal{R}$. Consequently, for every $i \in \{1\ldots, \ell\}$ it holds that $|R^* \cap R_{i}|\geq \beta$ as desired. 
Since the number of sets is bounded by $\ell\cdot \beta$, the running time is $\widetilde{O}(\ell \beta \Delta)$.
\end{proof}
Consequently, we obtain a deterministic algorithm for computing the center set $Z_i$ that satisfies the two properties.
\begin{lemma}\label{lem:det-centers-phasei}
A subset $Z_i \subseteq Z_{i-1}$ that satisfies properties (1) and (2) can be computed deterministically time $\widetilde{O}(f^{1-1/k} n^{1+1/k})$. 
\end{lemma}
\begin{proof}
Let $R=Z_{i-1}$, and for every $v \in V_i$ let $R_v=Z'_{i-1}(v)$. The set $Z_i$ is computed by applying the algorithm of Cor. \ref{cor:fhittingsets} on the sets $\{R_v, v \in V_i\}$ with the parameters $\beta=K_f$, $\Delta=1/p$ and $\ell=|V_{i-1}|\leq n$, let $Z_{i}=R^*$ where $R^*$ is the output hitting set. 
We have that $|Z_{i}|=|Z_{i-1}|/\Delta=|Z_{i-1}|\cdot p$ and in addition, $|Z'_{i-1}(v) \cap Z_i|\geq K_f$ as desired. By plugging the parameters, we get that the running time is $\widetilde{O}(kf\cdot n \cdot (n/f)^{1/k})$.
\end{proof}
\noindent We are now ready to complete the proof of Theorem \ref{thm:suplinear-time-det}.
\begin{proof}
The correctness of the algorithm is immediate as well as the size analysis, for completeness we sketch the slight modifications for deterministic version of phase $i$. 

Each vertex $v \in V_{i-1}$ is defined as $i$-clustered if $|\mathcal{P}^*_{i-1}(v)|\geq \Theta(\beta \cdot \Delta \log n)$ where $\beta=K_f$ and $\Delta=1/p$ for $p=(f/n)^{1/k}$ (the sampling probability in the randomized procedure). Letting $\ell=\Theta(\beta \cdot \Delta \log n)$, for each $i$-clustered vertex $v$, we consider the set $\mathcal{P}^*_{i-1}(v)=\{P_1,\ldots, P_q\}$ ordered in increasing order based on the weight of the last edges of the paths. We then restrict attention to the first $\ell$ paths and define its set $R_v=\{h(P_i) ~\mid~ i \in \{1,\ldots, \ell\}$. Applying Alg. \ref{cor:fhittingsets} provides the set $Z_i$ and the cluster-paths of $v$ are given by the $K_f$ nearest paths in $\mathcal{P}^*_{i-1}(v)$ rooted at $s \in Z_i$. This allows us to the define the sets $LE_{i-1}(v)$ as in Eq. (\ref{eq:LEiv}). Since the largest index $i_v$ of a path in $\mathcal{P}^*_{i-1}(v)$ taken into $\mathcal{Q}_i(v)$ is at most $\ell$, we have that $|LE_i(v)|=\widetilde{O}(f^{1-1/k}n^{1/k})$. The same bound clearly holds for $i$-unclustered vertices as well. 
\\ \\
\noindent \textbf{Running time.} As we omitted the random sampling of Step (1.2), Step (1) is now implemented in time $\widetilde{O}(f m)$. To see this observe that for every $v$, we iterate over the collection $k$-hop paths $\bigcup_{u \in N(v)}\mathcal{Q}_{i-1}(u)$. Step (2.1) is now implemented by applying Lemma \ref{lem:det-centers-phasei} using $\widetilde{O}(f^{1-1/k} n^{1+1/k})=\widetilde{O}(m)$.  Remaining steps are unchanged and therefore the total running time is dominated by Step 1 which is now implemented in $\widetilde{O}(f m)$ time. The theorem follows.
\end{proof}

\subsection{The Distributed Setting}
Ghaffari and Kuhn \cite{GhaffariK18} presented the first local deterministic algorithms for graph spanners in the congest model. Combining this result with the breakthrough network decomposition result by Rohzon and Ghaffari \cite{RozhonG20} provides a deterministic algorithm for $(2k-1)$ spanners in $\widetilde{O}(1)$ time. We show that using 
\cite{RozhonG20} and \cite{GhaffariK18}, the sequential algorithm for Thm. \ref{thm:suplinear-time-det} can be implemented in $\widetilde{O}(1)$ time. As explained above, we omit the step of sampling the sets $\mathcal{S}_i(v)$ and simply use the collection of $K_f$ paths $\mathcal{Q}_i(v)$. This is the reason for increasing the running time by a factor of $f$. We next show that the center set $Z_i$ can be computed using dtributed algorithms for the hitting-set problem. 

\begin{definition}[The Distributed Hitting-Set Problem, \cite{GhaffariK18}]\label{def:dist-hitting-set}
Consider a graph $G=(V,E)$ with two special sets of vertices $L, R \subseteq V$ with the following properties: each vertex $\ell \in L$ knows a set of vertices $R(\ell) \subseteq R$ where $|R(\ell)|=\Theta(\Delta \log n)$ such that $\dist_{G}(\ell,r)\leq T$ for every $r \in R(\ell)$. Here, $\Delta$ and $T$ are two given parameters in the problem. Moreover, there is a $T$-round congest algorithm that can deliver one message from each vertex $r \in R$ to all nodes $\ell \in L$ for which $r \in R(\ell)$. (The same message is delivered to all vertices in $L$.) The objective in the hitting set problem is to select a subset $R^* \subseteq R$ such that $R^* \cap R(\ell)\neq \emptyset$ and $|R^*|\leq |R|/\Delta$.
\end{definition}

\begin{theorem}[\cite{GhaffariK18,RozhonG20}]\label{thm:hitting-alg}
There is a deterministic algorithm that in $\widetilde{O}(T)$ rounds solves the hitting set problem.
\end{theorem}

\begin{lemma}\label{lem:centers-det-distributed}
A center set $Z_i \subseteq Z_{i-1}$ satisfying the two properties can be computed deterministically in $\widetilde{O}(f)$ congest rounds.
\end{lemma}
\begin{proof}
Set $\Delta'=\log n k \cdot f^{1-1/k}n^{1/k}$ and $\Delta=(n/f)^{1/k}$. 
Let $R=\{v \in V ~\mid~ |R_{i-1}(v)|=\Theta(\Delta')\}$ and set $L=Z_i$. 
We now slightly modify these sets to fit the distributed hitting-setting of \cite{GhaffariK18} as in Cor. \ref{cor:fhittingsets}. For every $v \in R$, we partition $R_{i-1}(v)$ into disjoint $q=20kf\log n$ sets $R^1_{i-1}(v), \ldots, R^q_{i-1}(v) \subset R_{i-1}(v)$ of roughly equal size $\Theta(\Delta'/q)=\Theta(\Delta)$. 

A vertex $v$ is defined as $i$-clustered if $|Z'_{i-1}(v)|\geq c\cdot\log n K_f \cdot 1/p$ where $p=(f/n)^{1/k}$. For every vertex $v$ define $R_v=Z'_{i-1}(v)$. Note that each $R_v \subseteq Z_{i-1}$. Next, we partition each $R_v$ into $K_f$ parts $R_{v,1}, \ldots, R_{v,K_f}$, each of size $\Theta(\log n/p)$. We then apply the algorithm of Theorem \ref{thm:hitting-alg} where each vertex simulates $K_f$ vertices. Naively this can be done in $\widetilde{O}(f)$ rounds, resulting in a hitting set $Z_i \subseteq Z_{i-1}$ such that: (i) $|Z_i|\leq |Z_{i-1}|\cdot p$ and (ii) 
$|R_v \cap Z_i|\geq K_f$ as desired.  
\end{proof}
\noindent We are now ready to complete the proof of the deterministic construction.
\begin{proof}[Det. Alg for Theorem \ref{thm:linear-time-opt-randomized-distributed}]
We focus on phase $i$ and show that it can be implemented in time $\widetilde{O}(f)$ rounds.
By Lemma \ref{lem:centers-det-distributed} it is sufficient to consider the implementation of Phase 1.
To compute the set $\mathcal{P}^*_i(v)$, it is sufficient for each $v \in V_{i-1}$ to receive the collection of $\mathcal{Q}_{i-1}(u)$ paths from each of its neighbors $u \in N(v)$. Since $|\mathcal{Q}_{i-1}(u)|=O(k f)$ and each path $P \in \mathcal{Q}_{i-1}(u)$ has at most $i-1$ edges, overall $|V(\mathcal{Q}_{i-1}(u))|=O(k^2 f)$. We note that one can shave some $k$ factors by a more delicate analysis. This allows each vertex to locally compute  $\mathcal{P}^*_i(v)$. Overall, the running time is dominated by Phase 1, and therefore takes $\widetilde{O}(f)$ rounds.
\end{proof}



\paragraph{Acknowledgment.} I am very much grateful to Greg Bodwin for preliminary insightful discussions. I am also grateful to Michael Dinitz for sparking my interest in this problem, back then at 2011, and even more so recently in his inspiring talk at ADGA 2021.

\bibliographystyle{alpha}
\bibliography{refs}

\appendix

\section{Missing Proofs}\label{sec:missproof}

\APPENDOBSCLUSTCONT
\APPENDOBSCLUSTINDUC
\APPENDCLUSTPROB
\APPENDCLUSTINV
\APPENDMONO

\end{document}